\documentclass[journal]{IEEEtran}
\newtheorem{definition}{Definition}
\newtheorem{theorem}{Theorem}
\newenvironment{proof}{\begin{IEEEproof}}{\end{IEEEproof}}
\usepackage{amssymb}
\newtheorem{example}{Example}  
\newtheorem{strategy}{Strategy}    
\newtheorem{lemma}{Lemma}  
\usepackage{multirow}  
\usepackage{hyperref}

\usepackage[switch]{lineno}   

\usepackage{amsmath}
\usepackage{algorithm}
\usepackage{algorithmic}

\ifCLASSINFOpdf
\usepackage[pdftex]{graphicx}
\DeclareGraphicsExtensions{.pdf,.jpeg,.png}
\else
\usepackage[dvips]{graphicx}
\DeclareGraphicsExtensions{.eps}
\fi

\ifCLASSINFOpdf
\else

\fi

\begin{document}
%

\title{HUOPM: High Utility Occupancy Pattern Mining}
%
%
        
\author{Wensheng Gan, 
        Jerry Chun-Wei Lin*, 
        Philippe Fournier-Viger,\\ 
        Han-Chieh Chao,~\IEEEmembership{Senior Member,~IEEE}
        and Philip S. Yu,~\IEEEmembership{Fellow,~IEEE}

\thanks{This research was partially supported by the Shenzhen Technical Project under JCYJ20170307151733005 and KQJSCX20170726103424709. Corresponding author: Jerry Chun-Wei Lin, Email: jerrylin@ieee.org}
\thanks{Wensheng Gan is with the School of Computer Science and Technology, Harbin Institute of Technology (Shenzhen), Shenzhen, China, and with the Department of Computer Science, University of Illinois at Chicago, IL, USA. (Email: wsgan001@gmail.com)}
\thanks{Jerry Chun-Wei Lin is with the Department of Computing, Mathematics, and Physics, Western Norway University of Applied Sciences, Bergen, Norway. (Email: jerrylin@ieee.org)}
\thanks{Philippe Fournier-Viger is with the School of Natural Sciences and Humanities, Harbin Institute of Technology (Shenzhen), Shenzhen, China. (Email: philfv@hitsz.edu.cn)}
\thanks{Han-Chieh Chao is with the Department of Electrical Engineering, National Dong Hwa University, Hualien, Taiwan (Email: hcc@ndhu.edu.tw)}
\thanks{Philip S. Yu is with the Department of Computer Science, University of Illinois at Chicago, IL, USA. (Email: psyu@uic.edu)}%

}

\markboth{Journal of \LaTeX\ Class Files,~Vol.~14, No.~8, August~2015}%
{Shell \MakeLowercase{\textit{et al.}}: Bare Demo of IEEEtran.cls for IEEE Journals}

\maketitle

\begin{abstract}
Mining useful patterns from varied types of databases is an important research topic, which has many real-life applications. Most studies have considered the frequency as sole interestingness measure for identifying high quality patterns. However, each object is different in nature. The relative importance of objects is not equal, in terms of criteria such as the utility, risk, or interest. Besides,  another limitation of frequent patterns is that they generally have a low occupancy, i.e., they often represent small sets of items in transactions containing many items, and thus may not be truly representative of these transactions. To extract high quality patterns in real-life applications, this paper extends the occupancy measure to also assess the utility of patterns in transaction databases. We propose an efficient algorithm named \underline{\textbf{H}}igh \underline{\textbf{U}}tility \underline{\textbf{O}}ccupancy \underline{\textbf{P}}attern \underline{\textbf{M}}ining (HUOPM). It considers user preferences in terms of frequency, utility, and occupancy. A novel \underline{\textbf{F}}requency-\underline{\textbf{U}}tility tree (FU-tree) and two compact data structures, called the utility-occupancy list and FU-table, are designed to provide global and partial downward closure properties for pruning the search space. The proposed method can efficiently discover the complete set of high quality patterns without candidate generation. Extensive experiments have been conducted on several datasets to evaluate the effectiveness and efficiency of the proposed algorithm. Results show that the derived patterns are intelligible, reasonable and acceptable, and that HUOPM with its pruning strategies outperforms the state-of-the-art algorithm, in terms of runtime and search space, respectively.
\end{abstract}


%
\IEEEpeerreviewmaketitle

\section{Introduction}
\IEEEPARstart{F}{requent} pattern mining (FPM) and association rule mining (ARM) \cite{agrawal1994fast,han2004mining,geng2006interestingness} are some of the important and fundamental data mining techniques to extract meaningful and useful information from massive amounts of data \cite{chen1996data,gan2017data}. FPM is the process of discovering frequent sets of items in transaction databases based on a user-specified minimum \textit{support} threshold. In recent decades, the task of frequent pattern mining has been extensively studied by mainly considering the frequency measure  for selecting patterns. Other properties and interestingness measures of frequent patterns and association rules have also been studied \cite{geng2006interestingness} to fulfill the need of discovering  more interesting patterns in databases, e.g., maximal frequent patterns \cite{gouda2005genmax}, all-confidence \cite{omiecinski2003alternative}, coherence \cite{omiecinski2003alternative}, and other constraints \cite{pei2001mining}.

In real-life applications, the importance of objects or patterns is often evaluated in terms of implicit factors such as the utility, interestingness, weight, risk or profit \cite{lin2016weighted,gan2018survey}. Hence, the knowledge that actually matters to the user may not be found using traditional FPM and ARM algorithms. To measure the utility of patterns, a utility-based mining framework called high utility pattern mining (HUPM) \cite{gan2018survey} was proposed, which considers the relative importance of items (\textit{item utility}). It has become an emerging research topic in recent years \cite{ahmed2009efficient,liu2012mining,tseng2010up,tseng2013efficient}. Chan et al. \cite{chan2003mining} first introduced the problem of utility-based pattern mining based on business objectives. Yao et al. \cite{yao2004foundational} then defined utility mining as the problem of finding profitable itemsets while considering both the purchase quantities of objects/items in transactions (\textit{internal utilities}) and their unit profits (\textit{external utilities}). In addition, the utility (i.e., importance, interest or risk) of each object/item can be predefined based on users' background knowledge or preferences. 

Recently, a study \cite{tang2012incorporating} has shown that considering the \textit{occupancy} of patterns is critical for many applications. The occupancy measure is used to ensure that  each pattern found represents a large part of transactions where it appears. This allows to find patterns that are more representative, and thus of higher quality. In many applications, occupancy is an interestingness factor to measure a pattern's completeness; and it is an indispensable complement to frequency (or support). 
However,  implicit factors such as the utility, interestingness, risk or profit of objects or patterns are ignored in \textit{occupancy pattern mining} \cite{tang2012incorporating}, and it ignores the fact that items/objects may appear more than once in transactions. On the other hand, HUPM does not assess the occupancy of patterns. As a result, the discovered patterns may be irrelevant or even misleading if they are not representative of the supporting transactions. For example, if a pattern occurs in many transactions but its actual occupancy is low, it may be inappropriate to use this pattern for recommendation. Hence, it is desirable to find patterns that are representative of the transactions where they occur. In particular, extracting patterns that occupy a large portion of the utility in their supporting transactions is critical to several applications. Recently, an algorithm called OCEAN was proposed to address the problem of high utility occupancy pattern mining by introducing the \textit{utility occupancy} measure \cite{shen2016ocean}. However, it fails to discover the complete set of high utility occupancy patterns (HUOPs) and also encounter several performance problems.

Therefore, this paper proposes an effective and efficient algorithm named \textbf{\underline{H}}igh \textbf{\underline{U}}tility \textbf{\underline{O}}ccupancy \textbf{\underline{P}}attern \textbf{\underline{M}}ining (\textbf{HUOPM}). The proposed algorithm extracts patterns based on users' interests, pattern frequency and utility occupancy, and considers that each item may have a distinct utility. The concept of \textit{utility occupancy} is adopted to evaluate the utility contribution of patterns in their supporting transactions. Thus, the concept of HUOP is quite different from the previous concept of high utility pattern. The previous study \cite{tang2012incorporating} shows that the \textit{occupancy} of patterns is very critical to some applications. Compared with \textit{occupancy}, \textit{utility occupancy} is suitable and more effective for pattern analysis in some real-life domains, for example: market basket analysis \cite{agrawal1994fast,han2004mining,tseng2013efficient}, print-area recommendation for Web pages \cite{tang2012incorporating}, Web click analysis \cite{ahmed2009efficient}, mobile service provider \cite{shen2016ocean}, and biomedical applications \cite{zihayat2017mining}. 

Consider the travel route recommendation for tourist to visit, eat, and spend time/money, HUOPM can successfully exploit the frequency and utility contribution ratio (w.r.t. utility occupancy) of a specific travel route. For another example, a mobile service provider can obtain some real-world data, and each record contains the traffic information that a customer spends on various mobile Apps. In general, each App has its unit utility (i.e., price, popular). In this case, the concept of \textit{utility occupancy} and the HUOPM model can be applied to discover the high qualified patterns (both frequent and high utility occupancy). After identifying the set of mobile Apps that users are interested in (e.g., download frequently, spend most of their money), this information can also be utilized to improve the service.  The major contributions of this paper are summarized as follows. 

\begin{itemize}
	\item A novel and effective HUOPM algorithm is proposed to address the novel research problem of mining high utility occupancy patterns with the utility occupancy measure. To the best of our knowledge, no prior algorithms address this problem successfully and effectively. 

	\item Two compact data structures called \underline{\textbf{U}}tility-\underline{\textbf{O}}ccupancy list (UO-list) and \underline{\textbf{F}}requency-\underline{\textbf{U}}tility table (FU-table), are designed to store the required information about a database, for mining HUOP. When necessary and sufficient conditions are met, information about a pattern can be directly obtained from the built UO-lists of its prefix patterns, and thus HUOPM can avoid repeatedly scanning the database. 
	\item Moreover, the concept of remaining utility occupancy is utilized to calculate an upper bound to reduce the search space. Based on the developed pruning strategies, the HUOPM algorithm can directly discover HUOPs from the designed frequency-utility tree using UO-lists by only scanning a database twice. 
	\item Extensive experiments have been conducted on both real-world and synthetic datasets to evaluate how effective and efficient the proposed HUOPM algorithm is. Results show that  HUOPM  can reveal the desired useful patterns, having a high utility occupancy and predefined frequency, and that the  proposed pruning strategies are effective at leading to a more compact search space. 
\end{itemize}

The rest of this paper is organized as follows. Related work is reviewed in Section \ref{sec:2}. In Section \ref{sec:background}, some key preliminaries are introduced and the addressed problem is defined. The proposed HUOPM algorithm and several pruning strategies are described in Section \ref{sec:4}. Then, to evaluate the effectiveness and efficiency of the HUOPM algorithm, results of extensive experiments comparing its performance with the state-of-the-art algorithm are provided in Section \ref{sec:experiments}. Finally, conclusions and future work are drawn in Section \ref{sec:conclusion}.

\section{Related Work}
\label{sec:2}
\subsection{Support-based Pattern Mining}
In numerous domains, web mining and data mining technologies provide powerful ways of discovering interesting and useful information in massive amounts of data. In recent decades, the task of support-based frequent pattern mining has been widely studied. Most algorithms have been designed to extract patterns using the frequency (support) measure as selection criterion  \cite{agrawal1994fast,han2004mining,geng2006interestingness}. Among them, Apriori \cite{agrawal1994fast} and FP-growth \cite{han2004mining} are two of the most well-known frequent pattern mining algorithms. Other properties and interestingness measures of frequent patterns and association rules have also been studied to fulfill the need for discovering more interesting patterns \cite{geng2006interestingness,omiecinski2003alternative,pei2001mining}

In the past, significant progress have been made concerning data mining with support and confidence \cite{geng2006interestingness}. Many algorithms have been proposed to discover high quality patterns from the binary transaction databases \cite{geng2006interestingness,omiecinski2003alternative,pei2001mining,luna2016speeding}. However, these approaches do not consider that objects/items may occur more than once in a transaction (e.g., may have non binary purchase quantities), and that objects/items are often not equally important to the user (e.g., may not have the same unit profit). Up to now, associations of boolean attributes only considered in traditional ARM, which only reflects the frequency of the presence or absence of an item in the database. Thus, the problem of quantitative association rule mining (QARM) has been studied \cite{srikant1996mining,hong1999mining,verlinde2005fuzzy}. In QARM, the attributes of item can be quantitative (e.g., age, income, purchase quantity) instead of the boolean value (0 or 1). However, QARM still does not reflect the other important factors of items, such as interest, price, risk or profit. In some real-life applications, the frequency (support) of a pattern may be an inappropriate measure to determine the importance of this pattern. 

\subsection{Utility-based Pattern Mining}
To address the limitations of support-based pattern mining and to extract high profitable patterns, the high-utility pattern mining (HUPM) task has been studied \cite{ahmed2009efficient,liu2012mining,tseng2010up,tseng2013efficient}. HUPM considers both the occur quantities and unit profits of objects/items rather than just considering their occurrence frequencies. A pattern is concerned as a high-utility pattern (HUP) if its utility is no less than the predefined minimum utility threshold.  Chan et al. \cite{chan2003mining} presented a framework to mine the top-\textit{k} closed utility patterns. Yao et al. \cite{yao2004foundational} then defined utility mining as the problem of discovering profitable patterns while considering both the purchase quantity of objects/items in transactions (\textit{internal utility}) and their unit profits (\textit{external utility}). They then introduced the mathematical properties of utility constraint to respectively reduce the search space and expected utility upper bounds \cite{yao2006mining}. The concept of HUPM is different from the support-based weighted frequent pattern mining (WFPM) \cite{lin2016weighted,tao2003weighted,lin2015rwfim}, because WFPM does not consider the quantity information of each object/item, and the range of weight is [0,1]. By taking both quantities and profits of items into consideration, HUPM can reveal more valuable patterns than frequent ones. Liu et al. then presented a transaction-weighted utilization (TWU) model \cite{liu2005two} to discover HUPs by adopting a transaction-weighted downward closure (TWDC) property. In recent years, developing algorithms for mining high utility patterns is an active research topic. Some recent algorithms for HUPM are IHUP \cite{ahmed2009efficient}, UP-growth \cite{tseng2010up}, UP-growth+ \cite{tseng2013efficient}, HUI-Miner \cite{liu2012mining}, d2HUP \cite{liu2012direct}, FHM \cite{fournier2014fhm}, HUP-Miner \cite{krishnamoorthy2015pruning}, and EFIM \cite{zida2017efim}. 

Different from the efficiency issue of HUPM, there are also a number of studies that focus on some interesting effectiveness issues of HUPM. For example, mining high utility patterns from different types of data (i.e., uncertain data \cite{lin2016efficient,lin2017efficiently}, temporal data \cite{lin2015efficient,lin2017two}, transaction data with negative unit profits \cite{lin2016fhn}, dynamic data \cite{lin2015incremental,lin2016fast2,lin2015fast,2gan2018survey}, and stream data \cite{yun2017efficient}), HUPM with various discount strategies \cite{lin2016fast}, non-redundant correlated HUPs \cite{gan2018extracting}, HUPM using multiple minimum utility thresholds \cite{lin2016efficient}, a condensed set of HUPs \cite{tseng2016efficient}, discriminative HUPs \cite{lin2017fdhup}, top-$k$ issue of HUPM \cite{tseng2016efficient}, and HUPM from big data \cite{lin2015mining}. All these HUPM algorithms discover high utility patterns based on the basic definitions of utility mining model by Yao et al. \cite{yao2004foundational,yao2006mining}. In recent years, there has been a growing interest in utility-oriented mining models for discovering different types of high utility patterns (e.g., itemsets, rules, sequences, and episodes) and profitable information. The comprehensive survey of this research field can be referred to \cite{gan2018survey,2gan2018survey}.

\subsection{High Quality Pattern Mining}

Although HUPM methods can evaluate the utility of patterns,  they do not assess the occupancy of patterns in their supporting transactions in terms of utility. In other words, HUPM methods do not consider how important the utility of patterns are compared to the utility of the transactions where they appear. For some real-life applications, such as pattern-based recommendation, it is desirable that interesting patterns should occupy a large portion of the transactions where they appear \cite{tang2012incorporating}. The above approaches fail to meet this requirement. To address this issue, a novel measure called \textit{utility occupancy} was proposed in the OCEAN algorithm \cite{shen2016ocean}. As a relative measure, utility occupancy ensures that a certain set of items is important to individual users. 

However, the OCEAN algorithm  suffers from two important drawbacks. First, the mining results derived by OCEAN are incomplete. The reason is that the exact utility information is incorrectly kept in the utility-list \cite{liu2012mining} structure using an inconsistent sorting order. As a result, OCEAN applies pruning strategies with incorrect information. Second, OCEAN is not efficient as it does not utilize the support property and utility occupancy property well to prune the search space. As a result, it performs poorly when the related parameters are set low. Hence, developing an effective and efficient algorithm to address these limitations is important. In this study, the problem of effectively mining the complete set of all high utility occupancy patterns from a database is discussed.


\section{Preliminaries and Problem Statement} 
\label{sec:background}

Let \textit{I} = \{\textit{i}$_{1}$, \textit{i}$_{2}$, $\ldots$, \textit{i$_{m}$}\} be a finite set of \textit{m} distinct items in a transactional quantitative database \textit{D} = \{\textit{T}$_{1}$, \textit{T}$_{2}$, $\ldots$, \textit{T$_{n}$}\}, where each quantitative transaction \textit{T$_{q}$} $ \in$ \textit{D} is a subset of \textit{I}, and has a unique identifier \textit{tid}. The total utility of all items in a transaction is named transaction utility and denoted as \textit{tu}. An itemset \textit{X} with \textit{k} distinct items \{\textit{i}$_{1}$, \textit{i}$_{2}$, $\ldots$, \textit{i$_{k}$}\} is called a \textit{k}-itemset. A database consisting of 10 transactions and 5 distinct items is shown in TABLE \ref{table:db}, which will be used as a running example.

\begin{table}[!htbp]
	\centering
	\small
	\caption{An example quantitative database.}
	\label{table:db}
	\begin{tabular}{|c|c|c|c|}
		\hline
		\textbf{\textit{tid}} & \textbf{Transaction (item, quantity)} & \textbf{\textit{tu}} \\ \hline
		$ T_{1} $ & 	\textit{a}:2, \textit{c}:4, \textit{d}:7  &  \$65 \\ \hline
		$ T_{2} $ & 	\textit{b}:2, \textit{c}:3   & \$37 \\ \hline
		$ T_{3} $ &	    \textit{a}:3, \textit{b}:2, \textit{c}:1, \textit{d}:2  &  \$38 \\ \hline
		$ T_{4} $ &	    \textit{b}:4, \textit{d}:3  &  \$11 \\ \hline
		$ T_{5} $ &	    \textit{a}:1, \textit{b}:3, \textit{c}:2, \textit{d}:5, \textit{e}:1 &   \$49  \\ \hline
		$ T_{6} $ &	     \textit{c}:2, \textit{e}:4  &  \$58 \\ \hline
		$ T_{7} $ &	     \textit{c}:2, \textit{d}:1 &  \$23 \\ \hline
		$ T_{8} $ &	     \textit{a}:3, \textit{b}:1, \textit{d}:2, \textit{e}:4  &  \$61 \\ \hline
		$ T_{9} $ &  	\textit{a}:2, \textit{c}:4, \textit{d}:1 &  \$59 \\ \hline
		$ T_{10} $	&	\textit{c}:3, \textit{e}:1 &  \$42 \\ \hline	
	\end{tabular}
\end{table}

\begin{definition}
	\label{def_1}
	\rm The number of transactions containing an itemset is said to be its occurrence frequency or \textit{support count} \cite{agrawal1994fast,han2004mining}. The support count of an itemset $X$, denoted as $ sup(X) $, is the number of \textit{supporting transactions} containing $X$. Let the set of transactions supporting an itemset $X$ denote as $ \varGamma_X$. A transaction $T_q$ is said to support an itemset $X$ if $ X \subseteq T_q $. Thus, $ sup(X) = |\varGamma_X| $. Let the user-specified minimum support threshold denote as $\alpha $, $ X $ is called a frequent pattern (\textit{FP}) in a database $D$ if $ sup(X) \geq \alpha \times |D| $. 
	
\end{definition}

\begin{definition}
	\label{def_2}
	\rm Each item $i_m$ in a database $D$ has a unit profit, denoted as $pr(i_m)$, which represents its relative importance to the user. Item unit profits are indicated in a user-specified profit table, denoted as \textit{ptable} = \{\textit{pr}$(i_{1})$, \textit{pr}$(i_{2})$, $\ldots$, \textit{pr$(i_{m})$}\}. The utility of an item $ i_{j} $ in a transaction $ T_{q} $ is defined as $ u(i_{j}, T_{q}) = q(i_{j}, T_{q})\times pr(i_{j}) $, in which $ q(i_{j}, T_{q})$ is the occur quantity of $ i_{j} $ in $ T_{q} $. The utility of an itemset/pattern \textit{X} in a transaction $ T_{q} $ is defined as $ u(X, T_{q}) = \sum _{i_{j}\in X\wedge X\subseteq T_{q}}u(i_{j}, T_{q}) $. Thus, the total utility of \textit{X} in a database \textit{D} is $	u(X) = \sum_{X\subseteq T_{q}\wedge T_{q}\in D} u(X, T_{q}) $.
\end{definition}

\begin{example} 
	In TABLE \ref{table:db}, assume that the unit profit of items $ (a) $ to $ (e) $ are  defined as \{$pr$(\textit{a}):\$7, $pr$(\textit{b}):\$2, $pr$(\textit{c}):\$11, $pr$(\textit{d}):\$1, $pr$(\textit{e}):\$9\}, respectively. Consider the itemsets $ (a) $ and $ (ab) $, their utilities in $T_3$ are $ u(a, T_3) $ = $ 3 \times \$7 $ = \$21, and $ u(ab, T_3) $  = $ 3 \times \$7  + 2 \times \$2 $ = \$21 + \$4 = \$25, respectively. Thus, their utilities in the database are  calculated as $ u(a) $ = $ u(a, T_1) $ + $ u(a, T_3) $ + $ u(a, T_5) $ + $ u(a, T_{8}) $ + $ u(a, T_{9}) $ = \$14 + \$21 + \$7 + \$21 + \$14 = \$77, and $ u(ab) $ = $ u(ab,T_3) $ + $ u(ab,T_5) $ + $ u(ab,T_8) $ = \$21 + \$13 + \$23 = \$57.
\end{example}

\begin{definition}
	\rm The transaction utility (\textit{tu}) of a transaction $ T_{q} $ is $
	tu(T_{q}) = \sum_{i_{j}\in T_{q}}u(i_{j}, T_{q}) $, where $i_j$ is the $j$-th item in $ T_q $.
\end{definition}

\begin{example} 
	$ tu(T_{1}) $ = $ u(a, T_{1}) $ + $ u(c, T_{1}) $ + $ u(d, T_{1}) $ = \$14 + \$44 + \$7 = \$65. The transaction utilities of transactions \textit{T}$ _{1} $ to \textit{T}$ _{10} $ are respectively calculated as \textit{tu}(\textit{T}$ _{1} $) = \$65, \textit{tu}(\textit{T}$ _{2} $) = \$37, \textit{tu}($T_{3} $) = \$38, \textit{tu}($T_{4}$) = \$11, \textit{tu}($T_{5}$) = \$49, \textit{tu}($T_{6}$) = \$58, \textit{tu}($T_{7}$) = \$23, \textit{tu}($T_{8}$) = \$61, \textit{tu}($T_{9} $) = \$59, and \textit{tu}(\textit{T}$ _{10} $) = \$42, as shown in TABLE \ref{table:db}.
\end{example} 

In the study \cite{tang2012incorporating}, a new interestingness measure called \textit{occupancy} was proposed to discover frequent patterns having a strong occupancy. A new concept named \textit{utility occupancy} \cite{shen2016ocean} is first introduced below. It is important to notice that the \textit{utility occupancy} concept is different from the original \textit{occupancy} concept presented in the DOFIA algorithm \cite{tang2012incorporating}.

\begin{definition}
	\label{def_woOfTq}
	\rm The utility occupancy of an itemset $X$ in a supporting transaction  $ T_q $ is denoted as $ uo(X, T_q) $, and defined as the ratio of the utility of $X$ in that transaction divided by the total utility of that transaction: 
	\begin{equation}
	uo(X, T_q) = \dfrac{u(X, T_q)}{tu(T_q)}.
	\end{equation}
\end{definition}

\begin{example} 
	Since $ tu(T_1)$ = \$65 and $ tu(T_3) $ = \$38, the utility occupancy of $ (ac) $ in $ T_1 $ is calculated as $ uo(ac,T_1) $ = \$58/\$65 $\approx $ 0.8923, and the utility occupancy of $ (ac) $ in $ T_3 $ is calculated as $ uo(ac,T_3) $ = \$32/\$38 $\approx $ 0.8421. 
\end{example} 

\begin{definition}
	\label{def_6}
	\rm The utility occupancy of an itemset $X$ in a database $ D $ is denoted as $ uo(X) $, and defined as: 
	\begin{equation}
	uo(X) = \dfrac{\sum_{X \subseteq T_q \wedge T_q \in D}uo(X,T_q)}{|\varGamma_X|} ,
	\end{equation}
	where $ \varGamma_X $ is the set of \textit{supporting transactions }of $X$ in $D$ (thus $ |\varGamma_X| $  is equal to the support of $X$ in $D$).
\end{definition}

Hence, utility occupancy can be used to evaluate how patterns contribute to the total utility of transactions where they appears in. In prior work \cite{tang2012incorporating}, the occupancy of a pattern in a transaction was defined as the ratio between the number of items in this pattern and the number of items in the transaction. The utility occupancy of a pattern generalizes this definition to consider that items may have distinct utility values, and is defined as the harmonic average of the occupancy utility values in all supporting transactions. Therefore, the two concepts, occupancy and utility occupancy, are different in nature. 

\begin{definition}
	\label{def_HUOP}
	\rm Given a minimum support threshold $ \alpha $ ($ 0 < \alpha \leq 1 $) and a minimum utility occupancy threshold $ \beta $ ($ 0 < \beta \leq 1 $), an itemset $X$ in a database $D$ is said to be a high utility occupancy pattern with high frequency and strong utility occupancy, denoted as \textit{HUOP}, if it satisfies the following two conditions: $ sup(X) \geq \alpha \times |D| $ and $ uo(X) \geq \beta $.
\end{definition}

\begin{example} 
	The utility occupancies of $ (a) $ and $ (ab) $ in TABLE \ref{table:db} are calculated as: $ uo(a) $ = ($ uo(a, T_1) $ + $ uo(a, T_3) $ + $ uo(a, T_5) $ + $ uo(a, T_{8}) $ + $ uo(a, T_{9}) $)/5 = (0.2154 + 0.5526 + 0.1429 + 0.3443 + 0.2373)/5 $\approx $ 0.2985, $ uo(ab) $ = ($ uo(ab,T_3) $ + $ uo(ab,T_5) $ + $ uo(ab,T_8) $)/3  = (0.6579 + 0.2653 + 0.3371)/3 = 0.4201. When $ \alpha $ and $ \beta $ are set to 30\%  and 0.30, the complete set of HUOPs in the running example database is: {$(c)$, $(e)$, $(ab)$, $(ac)$, $(ad)$, $(bc)$, $(bd)$, $(cd)$, $(ce)$, $(abd)$, $(acd)$}, as shown in TABLE \ref{table:patterns}. Clearly, the utility occupancy measure does not respect the \textit{downward closure} property of Apriori in FPM.
\end{example}

\begin{table}[!htbp]
	\centering
	\small
	\caption{The derived HUOPs.} 
	\label{table:patterns}
	\begin{tabular}{|c|c|c|c|c|c|}
		\hline
		\textbf{Pattern} & \textbf{\textit{sup}} & \textbf{\textit{uo}} & 	\textbf{Pattern} & \textbf{\textit{sup}} & \textbf{\textit{uo}} \\ \hline	
		$ (c) $ &	8  &	0.6468    &  $ (bd) $ &	4  &	0.3620  \\ \hline
		$ (e) $ &	4  &	0.4022	  &  $ (cd) $ &	5  &	0.6881  \\ \hline	
		$ (ab) $ &	3  &	0.4334    &  $ (ce) $ &	3  &	0.8776  \\ \hline
		$ (ac) $ &	4  &	0.8273    &  $ (abd) $ & 3  &	0.4959  \\ \hline
		$ (ad) $ &	5  &	0.3609    &  $ (acd) $ & 4  &	0.8972   \\ \hline	
		$ (bc) $ &	3  &	0.6554    & & & \\ \hline 
	\end{tabular}
\end{table}

Note that the minimum support count must be set such that $\alpha \times |D| > 1 $. Otherwise, each transaction (or an itemset) will have a utility occupancy value of 1. In the example of TABLE \ref{table:db}, if $\alpha$ is set to 0.1, then ($\alpha \times |D|$) = 1, and all ten transactions are HUOPs. For example, $T_1$ (i.e., the itemset ($acd$)) has a support count as 1 and a utility occupancy as 1.0, thus satisfying the two conditions of HUOP. Based on the above definitions, the problem statement of high utility occupancy pattern mining (HUOPM) is formulated as follows:

\textbf{Problem Statement.} Given a transaction database $ D $, a profit-table indicating the distinct profit of each item, a minimum support threshold $ \alpha $, and a minimum utility occupancy threshold $ \beta $. The problem of mining high utility occupancy patterns (which are both frequent and dominant in terms of high utility occupancy) is to discover the complete set of patterns that not only have a frequency no less than $ \alpha \times |D| $, but also have a utility occupancy no less than  $ \beta $.

\section{Proposed Algorithm for Mining HUOPs} 

\label{sec:4}
In the section, we propose two compact data structures called utility-occupancy list (UO-list) and frequency-utility table (FU-table) to maintain the utility occupancy information about a database. Then we utilize both the support and utility occupancy measures to prune the search space for mining the more interesting and useful high utility occupancy patterns.

\subsection{Search Space for Mining HUOPs}

According to previous studies, the search space for the pattern mining problem can be represented as a lattice structure \cite{pasquier1998pruning} or as a Set-enumeration tree \cite{rymon1992search}. Based on the final derived HUOPs of the given example, it can be observed that the well-known \textit{downward closure} property of Apriori does not hold for HUOPs. 
For example,  1-items $(a)$, $(b)$ and $(d)$ are not HUOPs, but their supersets \{$(ab)$, $(ac)$, $(ad)$, $(bc)$, $(bd)$, $(cd)$, $(abd)$, $(acd)$\} are HUOPs, as shown in TABLE \ref{table:patterns}. If no anti-monotone property is applied in HUOPM,  a huge number of candidates will need to be generated to obtain the actual HUOPs. It is thus a critical issue to design more suitable data structures and powerful pruning strategies to efficiently reduce the search space and to filter the number of unpromising patterns for mining HUOPs. 

\begin{definition}
	\label{def_10}
	\rm \textbf{(\textit{Total order $\prec$ on items})} Without loss of generality, assume that items in every transaction are sorted according to the lexicographic order. Furthermore, assume that the total order $\prec$ on items in the designed HUOPM algorithm adopts the support ascending order of items.
\end{definition}

Note that the total order $\prec$ used in the proposed HUOPM algorithm can be the support-based ascending or descending order, the TWU-based ascending or descending order, the lexicographic order, or any other total order on items. The processing order of patterns may, however, affect the mining efficiency of HUOPM, but does not affect the completeness and correctness of this algorithm. We assume that the support ascending order is adopted in the HUOPM algorithm. Notice that the efficiency when using different orders in HUOPM will be conducted and evaluated in the experiments. For example, the support counts of items in the running example are \{$sup$(\textit{a}):5, $sup$(\textit{b}):5, $sup$(\textit{c}):8, $sup$(\textit{d}):7, $sup$(\textit{e}):4\}. Thus, the support ascending order for five items is $\{sup(e) <$ $sup(a) \leq $ $sup(b) < sup(d)$ $< sup(c)\} $, and the total order $\prec$ on items is $ e \prec a \prec b \prec d \prec c $. 

\begin{definition}
	\label{def_11}
	\rm \textit{\textbf{(Frequency-utility tree, FU-tree)}}
	A frequency-utility tree (FU-tree) is a variant of sorted Set-enumeration tree using the total order $\prec$ on items. It contains the frequency and utility information.
\end{definition}

\begin{definition}
	\label{def_12}
	\rm \textbf{(\textit{Extension nodes in the FU-tree})} In the designed FU-tree, all child nodes of any tree node are called its extension nodes.
\end{definition}

Clearly, the designed FU-tree is a utility-based prefix-tree, and the complete search space can be traversed using a depth-first search (DFS) or breadth-first search (BFS), where each child node in the FU-tree is generated by extending its prefix (parent) node. Consider the running example with the adopted total order $\prec$ on items, the extension nodes of node $ (ab) $ are the itemsets $ (abd) $, $ (abc) $ and $ (abdc) $. Note that all the supersets of node $ (ab) $ are $ (eab)$, $(abd)$, $(abc)$, $(eabd)$, $(eabc)$, $(abde) $ and $ (eabde) $. Hence, the extension nodes of a tree node in a FU-tree are a subset of the supersets of that node.

\subsection{UO-List and FU-Table}

Two compact data structures, called utility-occupancy list (UO-list for short) and frequency-utility table (FU-table for short), are designed to keep essential information about patterns in the database. Note that the UO-list is different from the utility-list structure used in HUI-Miner \cite{liu2012mining}, FHM \cite{fournier2014fhm} and OCEAN \cite{shen2016ocean}. The occupancy information is not stored in utility-list, while UO-list keeps this information to quickly calculate the utility occupancy of a pattern using a join operation. Besides, a new concept called \textit{remaining utility occupancy} is introduced and applied to obtain an utility occupancy upper-bound, which will be presented in the next subsection. 

\begin{definition}
	\label{def_ruo}
	\rm The remaining utility occupancy of an itemset $X$ in a transaction  $ T_q $ is denoted as $ ruo(X,T_q) $, and defined as the sum of the utility occupancy values of all items appearing after $X$ in $T_q$ according to the total order $\prec$ , that is: 
	\begin{equation}
	ruo(X,T_q) = \dfrac{\sum_{ i_j \notin X \wedge X \subseteq T_q \wedge X \prec i_j }u(i_j,T_q)}{tu(T_q)}.
	\end{equation}
\end{definition}

For example, consider  $T_5$ and  itemsets  $(a)$ and $(ad)$ in TABLE \ref{table:db}. Since the total order $\prec$ on items is $ e \prec a \prec b \prec d \prec c $, we have that $ ruo(a,T_5)$ = $(u(b,T_5)$ + $u(d,T_5)$ + $u(c,T_5))/tu(T_5)$ = (\$6 + \$5 + \$22)/\$49  $ \approx $ 0.6735, and $ ruo(ad,T_{5})$ = $(u(c,T_5))/tu(T_5)$ = \$22/\$49 $ \approx $ 0.4490.

\begin{definition} 
	\rm \textbf{(\textit{Utility-occupancy list, UO-list}).} 
	Let $\prec$ be the total order on items from $I$. The \emph{utility-occupancy list} of an itemset $X$ in a database $D$ is a set of tuples corresponding to  transactions where $X$ appears. A tuple contains three elements $<$$\underline{\textbf{\textit{tid}}}, \underline{\textbf{\textit{uo}}}, \underline{\textbf{\textit{ruo}}}$$>$ for each transaction $T_{q}$ containing $X$. In each tuple, the $\underline{\textbf{\textit{tid}}}$ element is the transaction identifier of $T_{q}$; the $\underline{\textbf{\textit{uo}}}$ element is the utility occupancy of $X$ in $T_{q}$, w.r.t. $ uo(X,T_{q})$; and the $\underline{\textbf{\textit{ruo}}}$ element is defined as the remaining utility occupancy of $X$ in $T_{q}$, w.r.t. $ ruo(X,T_q) $.
\end{definition}

\begin{example} 
	Consider the running example and the defined total order $\prec$ ($ \{e \prec a \prec b \prec d \prec c \} $). The constructed UO-lists of five 1-itemsets (itemsets of length 1) are shown in Fig. \ref{fig:UO-list}.	Note that the UO-lists of all 1-itemsets are constructed after the HUOPM algorithm performs a single database scan.
\end{example}

\begin{figure}[!htbp]
	\centering
	\includegraphics[scale=0.50]{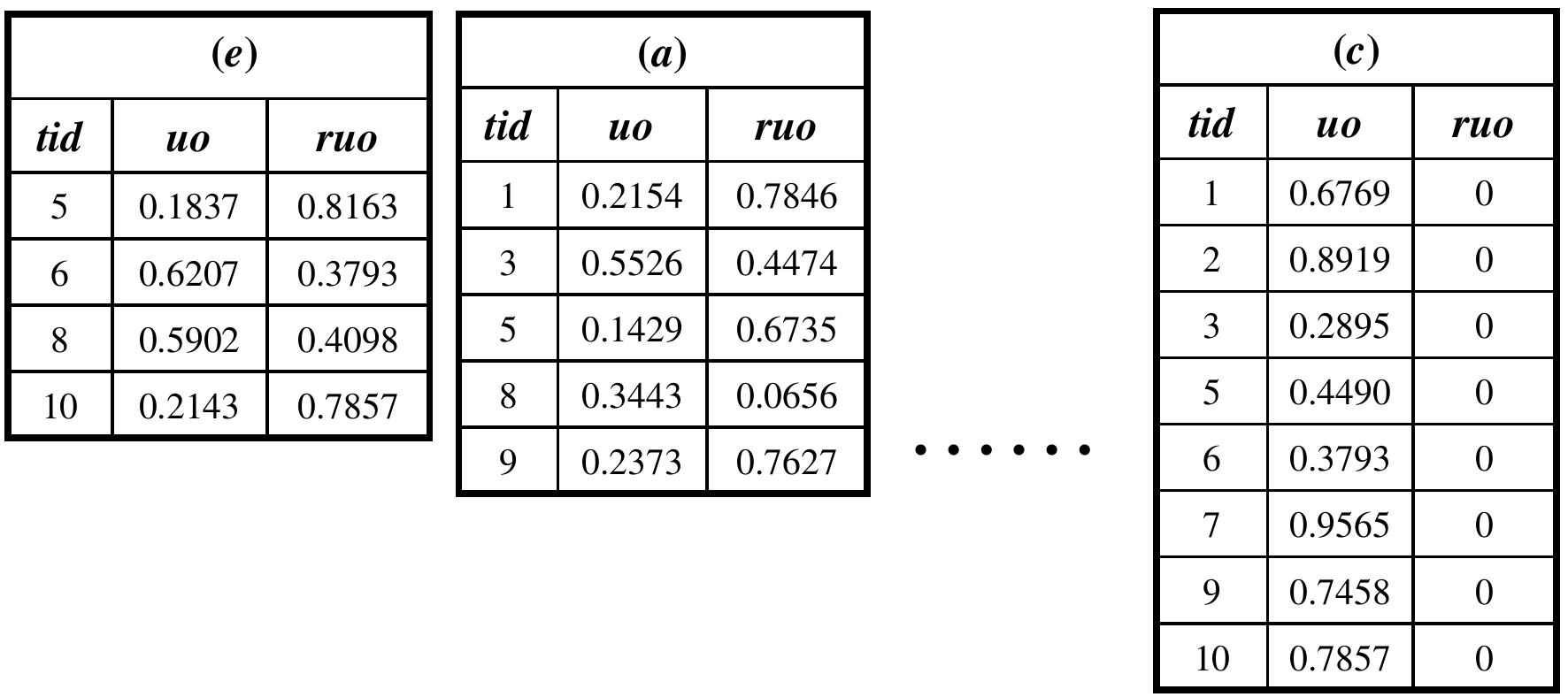}
	\caption{Constructed UO-lists of the five items.}
	\label{fig:UO-list}
\end{figure}

To discover HUOPs, the sum of the support and utility occupancy of a special pattern $X$ in a database $D$ can be efficiently calculated by adding the utility occupancies of all elements in the UO-list of $X$ (denoted as \textit{X.UOL}). Thus, the following information can be obtained from \textit{X.UOL}: (1) the name of $X$; (2) the set of transactions where $X$ appears (its support); (3) the sum of the utility occupancy of $X$ in $D$; and (4) the total remaining utility occupancy of $X$ in $D$. To explain how this useful information can be obtained from the constructed UO-lists, the following definitions are introduced by utilizing the UO-list structure. We further design a data structure called frequency-utility table (FU-table) by utilizing some useful properties of UO-list. The FU-table of an expected pattern is built after the construction of UO-list of this pattern, and it stores the following information.

\begin{definition}
	\label{def_15}
	\rm \textbf{(\textit{FU-table}).} A frequency-utility table (FU-table) of an itemset $ X $ contains four informations: the name of the itemset $ X $ (\textit{\textbf{\underline{name}}}), the support of $ X $ (\textit{\textbf{\underline{sup(X)}}}), the sum of the utility occupancies of $ X $ in database $ D $ (\textit{\textbf{\underline{uo(X)}}}), and the sum of the remaining utility occupancies of $ X $ in $ D $ (\textbf{\textit{\underline{ruo(X)}}}). Here, $ ruo(X) $ can be calculated as:
	\begin{equation}
	    ruo(X) = \dfrac{\sum_{X \subseteq T_q \wedge T_q \in D}ruo(X,T_q)}{|\varGamma_X|}. 
	\end{equation}
\end{definition}

The construction process of a FU-table is shown below. Consider an item $ (e) $ in TABLE \ref{table:db}, which appears in $T_5$, $T_6$, $T_8$, and $T_{10}$. The built UO-list of $ (e) $ is shown in Fig. \ref{fig:FWTableOfE}(a). The FU-table of item $(e)$ is constructed efficiently by using the support count, utility occupancy and remaining utility occupancy. They are calculated during the construction of the UO-list of $(e)$, such that $\{sup(e)$ = 4, $ uo(e)$ = (0.1837 + 0.6207 + 0.5902 + 0.2143)/4 = 0.4022, and  $ ruo(e)$ = (0.8163 + 0.3793 + 0.4098 + 0.7857)/4 = 0.5978$\} $, and the results of its FU-table are  $\{sup(e)$ = 4,  $ uo(e)$ = 0.4022, and  $ ruo(e)$ = 0.5978$\} $, as shown in Fig. \ref{fig:FWTableOfE}(b). Thus, the final built FU-tables of all items (1-itemsets) are shown in Fig. \ref{fig:FP1}.

\begin{figure}[!htbp]
	\centering
	\includegraphics[scale=0.55]{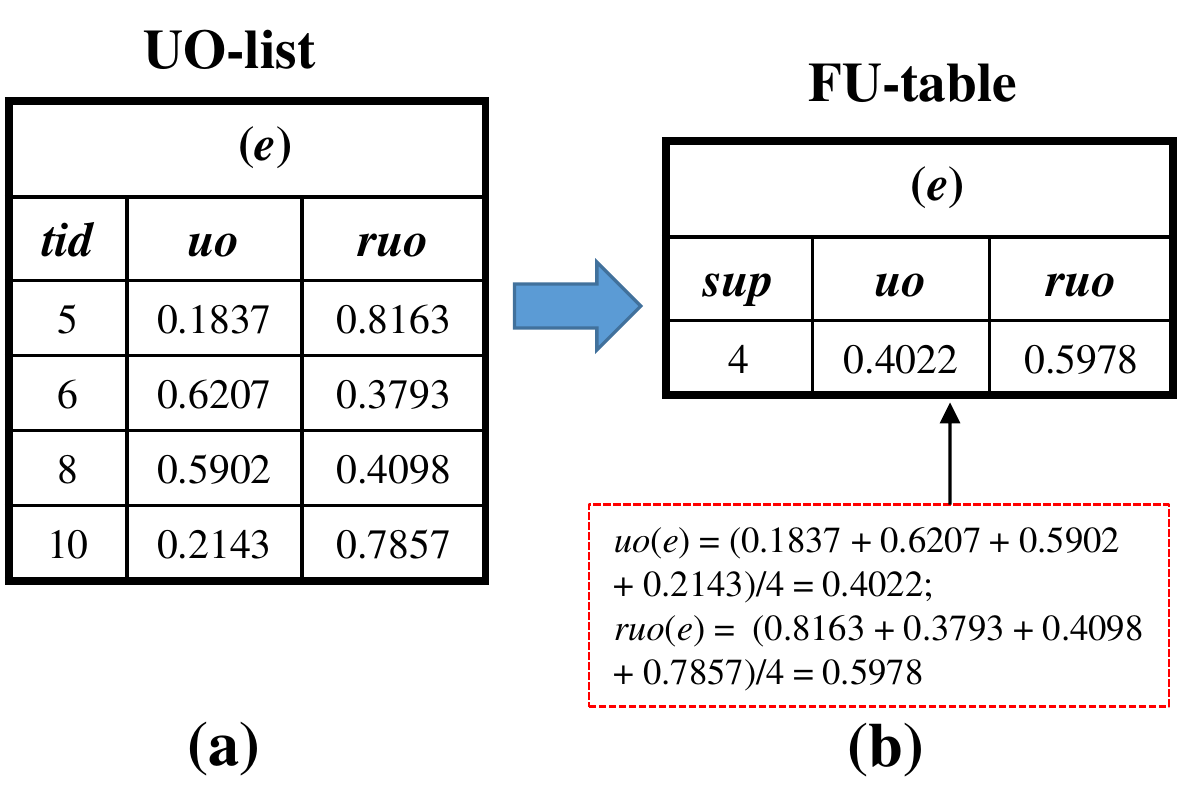}
	\caption{The UO-list and FU-table of item (e).}
	\label{fig:FWTableOfE}
\end{figure}

\begin{figure}[!htbp]
	\centering
	\includegraphics[scale=0.33]{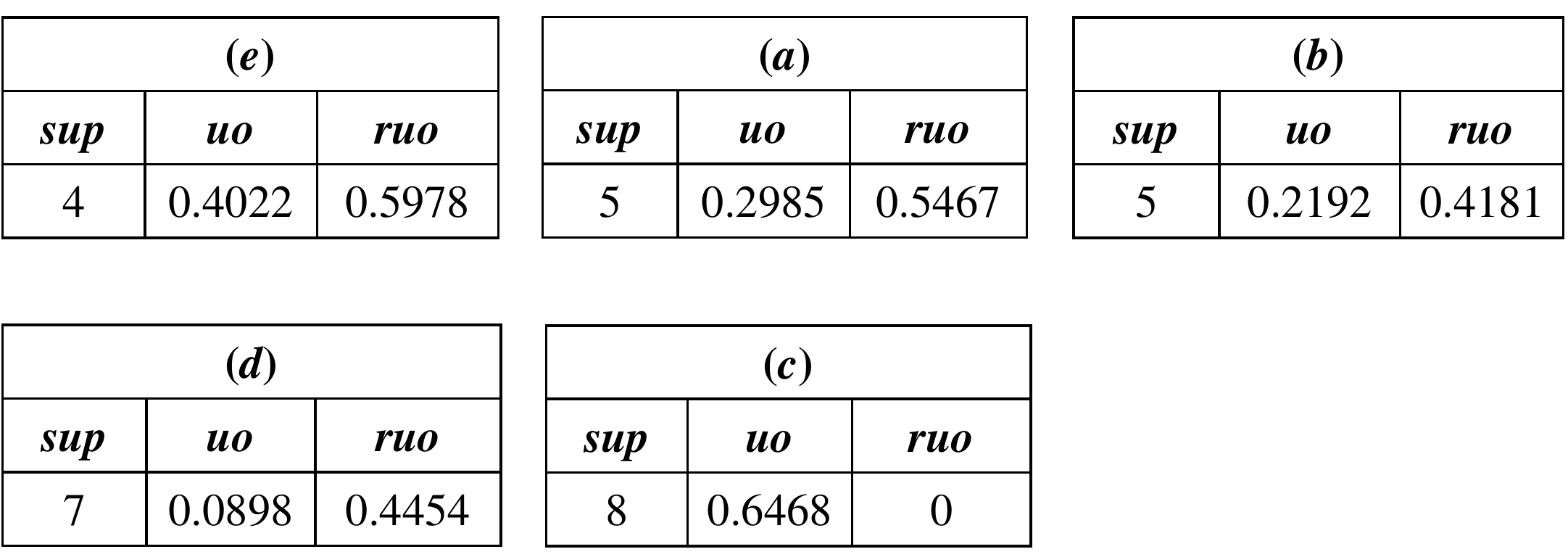}
	\caption{Constructed FU-tables of all 1-itemsets.}
	\label{fig:FP1}
\end{figure}

According to the two conditions $ sup(X) \geq \alpha \times |D| $ and $ uo(X) \geq \beta $ of HUOP, we only need to initially construct the UO-lists and FU-tables of  frequent 1-itemsets ($ FP^1 $). Then, for any $k$-itemset ($k \geq 2 $), the UO-list of this itemset can be calculated directly using the UO-lists of some of its subsets, without scanning the database. The construction procedure of the UO-list and FU-table of a $k$-itemset is shown in Algorithm 1. Given an itemset $X$, and two itemsets $X_a$ and $X_b$ which are extensions of $X$ obtained by adding two distinct items $a$ and $b$ to $X$, respectively. The construction procedure takes as input the UO-lists of $X$, $X_a$ and $X_b$, and outputs the UO-list and FU-table of the itemset $ X_{ab} $. Note that \textit{X$_{ab}$.UOL} denote the UO-list of $ X_{ab} $, and \textit{X$_{ab}$.FUT} denote the FU-table of $ X_{ab} $.

\begin{algorithm}
	\caption{Construction($ X $, $ X_{a} $, $ X_{b} $)}
	\begin{algorithmic}[1]
		\STATE set  $ X_{ab}.UOL \leftarrow \emptyset $, $ X_{ab}.FUT\leftarrow \emptyset $; 
		\FOR {each tuple $ E_{a}\in X_{a}.UOL $}	
		\IF {$ \exists E_{a}\in X_{b}.UOL \wedge E_{a}.tid == E_{b}.tid $}
		\IF{$ X.UOL \neq \emptyset $}			
		\STATE search for $ E\in X.UOL, E.tid = E_{a}.tid $;			

		\STATE $E_{ab}$ $\leftarrow$ $<E_{a}.tid$, $E_{a}.uo$ + $E_{b}.uo$ - $E.uo$, $E_{b}.ruo>$;
		\STATE $ X_{ab}.FUT.uo $ += $ E_{a}.uo + E_{b}.uo - E.uo $;
		\STATE $ X_{ab}.FUT.ruo $ += $ E_{b}.ruo $;
		\ELSE
		
		\STATE $E_{ab} \leftarrow <E_{a}.tid, E_{a}.uo + E_{b}.uo, E_{b}.ruo>$;
		\STATE $ X_{ab}.FUT.uo $ += $ E_{a}.uo + E_{b}.uo $;
		\STATE $ X_{ab}.FUT.ruo $ += $ E_{b}.ruo $;
		\ENDIF
		\STATE $ X_{ab}.UOL \leftarrow X_{ab}.UOL \cup E_{ab} $;
		\STATE $  X_{ab}.FUT.sup $ ++;	
		\ELSE		
		\STATE $  X_{a}.FUT.sup $ - -;
		\IF{$ X_{a}.FUT.sup < \alpha \times |D| $}
		\STATE \textbf{return} \textit{null};	
		\ENDIF
		\ENDIF
		\ENDFOR
		\STATE \textbf{return} $ X_{ab} $\
	\end{algorithmic}
\end{algorithm}

In Algorithm 1, Lines 17 to 20 show the developed pruning strategy named \textit{remaining support strategy}, which will be described later. It is important to notice that the construction of the UO-list and FU-table of a $k$-itemset ($ k \geq 3 $, Lines 5 to 8) is different than that of a $2$-itemset ($ k =2  $, Lines 10 to 12). Details of the difference are described in the following content. The $ uo$ of $X_{ab}$ is the sum of the $ uo$ associated with $tid$ in the UO-lists of $X_{a}$ and $X_{b}$. Suppose $a$ is before $b$, and then the $ ruo$ of $X_{ab}$ is assigned as the $ ruo$ associated with $tid $ in the UO-list of $X_{b}$. However, when calculating the $uo$ and $ruo$ for a $(k)$-itemset ($ k \geq 2 $), the part of $uo$ is different, thus it needs to subtract the $uo$ of the common part between $X_{a}$ and $X_{b}$. Generally, to calculate the $uo$ of $\{i_1 $\ldots$ i_{(k-2)}i_{(k-1)}i_k\} $ in $tid$, the following formula holds: $uo(\{i_1 \dot i_{(k-2)}i_{(k-1)}i_k\},tid) $ = $uo(\{i_1 $\ldots$ i_{(k-2)}i_{(k-1)}\}, tid) $ + $ uo(\{i_1 $\ldots$ i_{(k-2)}i_k\}, tid) $ - $ uo(\{i_1 $\ldots$ i_{(k-2)}\}, tid) $.
For example, in transaction $T_5$, $uo(ac, T_5)$ = $uo(a, T_5)$ + $uo(c, T_5)$ = 0.1429 + 0.4490 = 0.5919, $uo(ae, T_5)$ = $uo(a, T_5)$ + $uo(e, T_5)$ = 0.1429 + 0.1837 = 0.3266, and $uo(ace, T_5)$ = $uo(ac, T_5)$ + $uo(ae, T_5)$ - $uo(a, T_5)$ = 0.5919 + 0.3266 - 0.1429 = 0.7756. We can not miscalculate the $uo$ of $\{ace\}$ in $T_5$ such as the sum of the $uo$ of $\{ac\}$ and $\{ae\}$ in $T_3$, since they contain the $uo$ of $\{a\}$ in $T_5$ twofold.

The procedure can be easily implemented since a set of the UO-lists of the ($k$-1)-itemsets ($ k \geq 2 $) has been built before constructing the UO-list of a $(k)$-itemset ($ k \geq 2 $) w.r.t. the extension node.

\subsection{Upper Bound on Utility Occupancy}

Using the proposed UO-list and FU-table structures, the actual utility occupancy value of a pattern can be calculated exactly. But a crucial question is: is it necessary to construct the UO-lists and FU-tables of all patterns in a FU-tree? Since this would be very expensive in terms of runtime and memory usage, some pruning conditions should be developed to decide whether a subtree must be spanned. In a Set-enumeration tree \cite{rymon1992search}, the complete search space of $I$ (where $m$ is the number of items in $I$) contains $2^m$ patterns (by systematically enumerating all subsets of $I$ with the total order $\prec$, that is all possible patterns). As mentioned before, the \textit{downward closure} property does not hold for the utility occupancy. Thus, it is a crucial challenge to design powerful technologies to prune the search space and filter unpromising patterns early, especially to handle large-scale databases.

Without actually generating all possible high utility occupancy patterns, can we derive an upper bound $ \hat{\phi}(X) $  on the utility occupancy of patterns in a subtree rooted at node representing an itemset $X$? Unfortunately, it is not an easy task since the utility occupancy and the supporting transactions of patterns are unknown until these patterns are processed. 
Hence, inspired by the OCEAN algorithm \cite{shen2016ocean}, we try to develop an upper bound on the utility occupancy by utilizing the UO-list structure. If $ \hat{\phi}(X) $ is less than the minimum utility occupancy threshold, we can safely discard this subtree from further consideration. 

\begin{definition}
	\label{def_SetOfExtension}
	\rm Let there be an itemset $X$ and a transaction (or an itemset) $T$ such that $ X \subseteq T$, and that the set of all items in $T$ that are not in $X$ is denoted as $T \setminus X $. In addition, the set of all items appearing after $X$ in $T$ according to the $ \prec $ order is denoted as $T/X$. Thus, $ T/X \subseteq T  \setminus X $.  
\end{definition}

\begin{lemma}
	\label{lemma1}
	\rm Let there be a subtree rooted at $X$,  $ \varGamma_X $ be the supporting transactions of $X$. Then, for any possible high utility occupancy itemset $W$ in the subtree, we have:
	\begin{equation}
	uo(W) \leq \dfrac{\sum_{W \subseteq T_q \wedge T_q \in D}(uo(X,T_q) + ruo(X,T_q))}{|\varGamma_W|} .
	\end{equation}
	
\end{lemma}

\begin{proof}
	 Since $W$ is a $k$-extension of $ X$, $ \varGamma_W \subseteq \varGamma_X $, we have $ (W - X) = (W/X) $. Thus, in each transaction $T_q$, $ uo(W / X,T_q) \leq uo(T_q / X,T_q) $.	The following relationships can be obtained. For each transaction $T_q$ in $D$, since\\
\begin{tabbing}		
	$ uo(W)$ \= $= \dfrac{\sum_{W \subseteq T_q \wedge T_q \in D}uo(X,T_q)}{|\varGamma_W|} $\\
	\>$ = \dfrac{\sum_{W \subseteq T_q \wedge T_q \in D}(uo(X,T_q) + uo(W-X,T_q))}{|\varGamma_W|} $\\
	\>$ \leq \dfrac{\sum_{W \subseteq T_q \wedge T_q \in D}(uo(X,T_q) + uo(T_q/X,T_q))}{|\varGamma_W|} $\\ 
	
	$\Longrightarrow  uo(W) \leq	\dfrac{\sum_{W \subseteq T_q \wedge T_q \in D}(uo(X,T_q) + uo(T_q/X,T_q))}{|\varGamma_W|} $.\\
	$\Longrightarrow  uo(W) \leq	\dfrac{\sum_{W \subseteq T_q \wedge T_q \in D}(uo(X,T_q) + ruo(X,T_q))}{|\varGamma_W|} $.\\
\end{tabbing}				 	
\end{proof}

The supporting transactions of a processed node (itemset) $X$ is denoted as $ \varGamma_X $, but the supporting transactions of any possible high utility occupancy pattern $W$ extending $X$ w.r.t. $ \varGamma_W $, is unknown although $ \varGamma_W \subseteq \varGamma_X $.
The set $ \varGamma_W $ is unknown until $W$ is processed. Assume that we now explore the node $ (ea) $. After constructing the UO-list of $ (ea) $, we can obtain its support count and utility occupancy. Unfortunately, we do not know the related support count $ \varGamma_W $ of any of its extension nodes $W$. Because of Inequality (5), it is difficult to obtain an upper bound for $W$ without knowing  $ \varGamma_W $. The concept of HUOP indicates that every HUOP should be supported by $ \alpha \times |D| $ transactions in the database. By utilizing this property, we further develop the following theorems to obtain an upper bound for a subtree rooted at a processed node in the FU-tree.

\begin{lemma}
	\rm Let there be a minimum support threshold $ \alpha $, a subtree rooted at $X$, and $ \varGamma_X $ be its supporting transactions. For any pattern $W$ in the subtree, an upper bound on the utility occupancy of  $W$ is:
	\begin{equation}
	\hat{\phi}(W)  = \dfrac{\sum_{top \alpha \times |D|, T_q \in \varGamma_X}(uo(X,T_q) + ruo(X,T_q))}{|\alpha \times |D||}.
	\end{equation}

	\begin{equation}
	\hat{\phi}(W) \geq uo(W).
	\end{equation}
\end{lemma}

\begin{proof}
	Note that $ \varGamma_W \subseteq \varGamma_X $ and $ \varGamma_W $  is unknown, we calculate  $ (uo(X,T_q) + ruo(X,T_q)) $  for all transactions in $ \varGamma_X $ and put into a vector set (denoted as $ V_{occu} $), then sort $ V_{occu} $ in descending order (denoted as $ V_{occu}^{\downarrow} $). Since the average of top $k$ ($ 0 < k \leq |\varGamma_W| $) values of vector $ V_{occu}^{\downarrow} $ is an upper bound of the average of total $ |\varGamma_W| $ values of $ \dfrac{\sum_{W \subseteq T_q \wedge T_q \in D}(uo(X,T_q) + ruo(X,T_q))}{|\varGamma_W|} $. 
	Since a high utility occupancy pattern should be supported by at least  $ \alpha \times |D| $ transactions, we have $ \alpha \times |D| \leq k  \leq |\varGamma_W|\leq |\varGamma_X| $. We have,
	\begin{tabbing}	
		$ uo(W) \leq \dfrac{\sum_{W \subseteq T_q \wedge T_q \in D}(uo(X,T_q) + ruo(X,T_q))}{|\varGamma_W|} $ \\ 
		
		$\Longrightarrow  uo(W) \leq	\dfrac{\sum_{top k, T_q \in \varGamma_X}\{uo(X,T_q) + ruo(X,T_q)\}^{\downarrow}}{|\varGamma_W|} $\\
		$\Longrightarrow  uo(W) \leq	\dfrac{\sum_{top \alpha \times |D|, T_q \in \varGamma_X}\{uo(X,T_q) + ruo(X,T_q)\}^{\downarrow}}{|\alpha \times |D||} $\\
		$\Longrightarrow  uo(W) \leq	\hat{\phi}(W) $.\\	
		
	\end{tabbing}	
	Thus, given a minimum support threshold $ \alpha $, we can directly calculate an upper bond $ \hat{\phi}(W) $ on utility occupancy of a subtree which rooted at a processed node $X$.
\end{proof}

\subsection{Proposed Pruning Strategies}
In this section we will show how to efficiently prune the search space  using an upper bound on the utility occupancy and the support count. Two properties named \textit{global downward closure} property and \textit{partial downward closure} property can be obtained below.

\begin{lemma}
	\label{lemma_support}
	\rm The complete search space of the addressed HUOPM problem can be represented by a FU-tree where items are sorted according to the support ascending order on items.
\end{lemma}

\textit{Proof.} According to the studies in \cite{rymon1992search}, the complete search space for mining HUOPs can be presented as a Set-enumeration tree.

\begin{theorem}
	\label{theorem_GDC}
	\rm \textbf{(\textit{Global downward closure property in the FU-tree})} 
	In the designed FU-tree, if a tree node is a FP, its parent node is also a FP. Let $ X^k $ be a $k$-itemset (node) and its parent node be denoted as $ X^{k-1}$, which is a ($k$-1)-itemset. The relationship $ sup(X^k) \leq sup(X^{k-1}) $ holds.
\end{theorem}

\begin{proof}
	According to the well-known Apriori property~\cite{agrawal1994fast}, the relationship $ sup(X^k) \leq sup(X^{k-1}) $ exists. Thus, in the FU-tree, the \textit{global downward closure} property holds. 
\end{proof}

\begin{theorem}
	\label{theorem_CDC}
	\rm \textbf{(\textit{Partial  downward closure property in the FU-tree})} 
	In the designed FU-tree, let $ X^k $ be a $k$-itemset (node) and its parent node be denoted as $ X^{k-1} $, a ($k$-1)-itemset. The relationship indicating that the upper bound on utility occupancy of any node in a subtree is no greater than that of its parent node always holds, that is $ \hat{\phi}(X^k) \leq \hat{\phi}(X^{k-1}) $.
\end{theorem}

\begin{proof}
	 Since the UO-list of $ X^k $  is constructed by joining the UO-list of $ X^{k-1} $ with the one of a sibling of $ X^{k-1} $, we have $ \varGamma_{X^k} \subseteq \varGamma_{X^{k-1}} $. The average of the top $ k $ ($ 0 < k \leq |\varGamma_{X^k}| $ values in the vector $ V_{occu}^{\downarrow} $ of $ X^k $ is no greater than the average of top $ k$ ($ 0 < k \leq |\varGamma_{X^{k-1}}| $ values of vector $ V_{occu}^{\downarrow} $ of $ X^{k-1} $. According to Lemma 2, the relationship between upper bounds on the utility occupancy of $ X^k $ and $ X^{k-1} $ is: $ \hat{\phi}(X^k) \leq \hat{\phi}(X^{k-1}) $. As a result, if $ \hat{\phi}(X^{k}) \geq \beta $, then $ \hat{\phi}(X^{k-1}) \geq \hat{\phi}(X^{k}) \geq \beta$. Conversely, if $ \hat{\phi}(X^{k-1}) < \beta$, then $ \hat{\phi}(X^{k}) \leq \hat{\phi}(X^{k-1}) < \beta$. In other words, the obtained upper bound on the utility occupancy of patterns satisfies the \textit{partial downward closure} property. Thus, in the developed FU-tree, the \textit{partial downward closure} property holds.
\end{proof}

It is noteworthy that   ``\textit{partial}'' here indicates that the upper bound on the utility occupancy is conditional anti-monotone but not general anti-monotone. In other words, the \textit{partial downward closure} property holds for a tree node and its descendants, but does not hold for supersets which are not in that subtree. To further enhance the pruning effect of the proposed HUOPM algorithm, the derived upper bound can be utilized with this anti-monotone  property. We next present four novel pruning strategies, which were not used in the state-of-the-art OCEAN algorithm \cite{shen2016ocean}.

\begin{strategy}
	In the designed FU-tree and by considering the defined total order $\prec$, if a tree node $X$ has a support count less than ($ \alpha \times |D| $), then any nodes containing $X$ (i.e., all supersets of $X$) can be directly pruned and do not need to be explored.
\end{strategy}

\begin{example} 
	In the running example, assume that the parameters are set to $ \alpha $ = 30\% and  $ \beta $ = 0.3. Using the UO-lists of $(e)$ and $(a)$, we can construct the UO-list of $(ea)$. Since $ sup(ea) = 2 < (\alpha  \times 10  = 30\% \times 10 = 3$), all extension nodes of $(ea)$ are not HUOPs and can be directly pruned.
\end{example}

\begin{strategy}
	In the designed FU-tree, considering the defined total order $\prec$, if the upper bound on the utility occupancy of a tree node $X$ is less than $ \beta $, then any nodes in the subtree rooted at $X$ w.r.t. all extensions of $X$ can be directly pruned and do not need to be explored.
\end{strategy}

\begin{example} 
	After constructing the UO-list of $ (ab) $, we can obtain its utility occupancy (\textit{uo}) and remaining utility occupancy (\textit{ruo}) in each supporting transaction. According to inequality (5), the top 3 values in the sorted vector $ V_{occu}^{\downarrow} $ of $ (ab) $ are considered. Since the UO-list of $(ab)$ is $\{(T_3$, 0.6579, 0.3421), $(T_5$, 0.2653, 0.5510), $(T_8$, 0.3771, 0.0328)$\}$, $ V_{occu}^{\downarrow} $ of $ (db) $ =  $\{(T_3$, 1.0000), $(T_5$, 0.8163), $(T_8$, 0.4099)$\}$. Thus, the utility occupancy upper bound of ($ab$) is (1.0000 + 0.5380 + 0.4099)/3 $\approx$ 0.7421 $>$ 0.3, and $ sup(ab)$ = 3. Thus, its extension nodes  ($abd$), ($abc$) and ($abdc$) may be HUOPs, and hence should be explored.		
\end{example} 

\begin{strategy}
	During the construction of the UO-list of a tree node $X_{ab}$ by using $X_{a}$ and $X_{b}$, if the remaining support of $X_{a}$ for constructing $X_{ab}$ is less than $ \alpha \times |D|$, the support of $X_{ab}$ will also be less than $ \alpha \times |D|$,  $X_{ab}$ is not a FP, and also not a HUOP. Then the construction procedure returns \textit{null}, as shown in Algorithm 1 Lines 18 to 20. 
\end{strategy}

\begin{strategy}
	After calling the construction procedure to build the UO-list of a tree node $X$, if \textit{X.UOL} is empty or  $ sup(X) \leq \alpha \times |D|$, $ X$ is not a HUOP, and none of its subtree nodes is a HUOP. Then, \textit{X.UOL} is not added to the set of extension UO-lists of $X$. 
\end{strategy}

\begin{example} 
	In the running example, the node $(ea)$ has a support count of 2. Its subtree nodes can be pruned by Strategy 4. In this case, Strategy 4 has the same effect as Strategy 1. 
\end{example}

\subsection{Proposed HUOPM Algorithm}

Based on two compact data structures (UO-list and FU-table), an upper bound on utility occupancy and the above pruning strategies, the proposed HUOPM algorithm can be designed below. Algorithm 2 shows the pseudo-code of the designed HUOPM algorithm. It takes four parameters as input, including a transactional database $D$, a predefined profit table \textit{ptable}, a user-specified minimum support threshold $ \alpha $ ($ 0 < \alpha \leq 1 $), and a minimum utility occupancy threshold  $ \beta $ ($ 0 < \beta \leq 1 $). HUOPM first scans the database to calculate the $sup(i)$ of each item $ i \in I $ and the $tu$ value of each transaction (Line 1), then finds the set  $ I^* $ w.r.t. $FP^1$ (Line 2, Strategy 1), and sorts it according to the total order $\prec$ (Line 3). After that, the algorithm scans $D$ again to construct the UO-list and FU-table of each item $ i \in I^* $ (Line 4). At last, HUOPM recursively applies the designed \textbf{\textit{HUOP-Search}} function for mining HUOPs without generating candidates and without scanning the database repeatedly (Line 5).

\begin{algorithm}
	\label{HUOPM-algorithm}
	\caption{HUOPM (\textit{D}, \textit{ptable}, $\alpha$, $\beta$)}
	\begin{algorithmic}[1]		
		\STATE scan $D$ to calculate the $sup(i)$ of each item $ i \in I $ and the $tu$ value of each transaction;
		\STATE find $ I^* \gets \left\{  i \in I | sup(i) \geq \alpha \times |D| \right\} $, w.r.t. $ FP^1 $;
		\STATE sort $ I^* $ in the designed total order $ \prec $;
		\STATE using the total order $ \prec $, scan $D$ once to build the UO-list and FU-table for each 1-item $ i\in I^*$;
		\STATE \textbf{call \textit{HUOP-Search}}($\phi, I^*, \alpha, \beta $).		
		\STATE \textbf{return} \textit{HUOPs}		
	\end{algorithmic}
\end{algorithm}

\begin{algorithm}
	\label{HUOP-Search procedure}
	\caption{HUOP-Search ($X$, $\textit{extenOfX}$, $\alpha$, $\beta$)}
	\begin{algorithmic}[1]			
		\FOR {each itemset $ X_{a}\in $ $ \textit{extenOfX} $}	
		\STATE obtain $ sup(X_a) $ and $ uo(X_a) $  from the built $ X_{a}.FUT $;
		\IF{$ sup(X_a) \geq \alpha \times |D| $}
		
		\IF{$ uo(X_a)\geq \beta $}			 
		\STATE $ HUOPs\leftarrow HUOPs\cup X_{a} $;	
		\ENDIF
		
		\STATE	$ \hat{\phi}(X_a) \leftarrow  \textbf{\textit{UpperBound}}(X_a.UOL, \alpha) $;
		\IF{$ \hat{\phi}(X_a) \geq \beta $}
		
		\STATE $ \textit{extenOfX}_{a}\leftarrow  \emptyset $;
		\FOR {each $ X_{b}\in \textit{extenOfX} $ that $ X_{a} $   $ \prec $  $ X_{b} $}
		
		\STATE $ X_{ab}\leftarrow X_{a} \cup X_{b} $;
		
		\STATE call $ \textbf{\textit{Construct}}(X, X_{a}, X_{b}) $;
		\IF{$ X_{ab}.UOL \not= \emptyset \wedge sup(X_{ab}) \geq \alpha \times |D| $}
		
		\STATE $ \textit{extenOfX}_{a}\leftarrow \textit{extenOfX}_{a}\cup X_{ab}.UOL $;		
		\ENDIF			
		\ENDFOR		 						  		
		\STATE \textbf{call \textit{HUOP-Search}}$\boldmath{(X_{a}, \textit{extenOfX}_{a}, \alpha, \beta)}$;	
		\ENDIF
		\ENDIF
		\ENDFOR
		\STATE \textbf{return} \textit{HUOPs}
	\end{algorithmic}
\end{algorithm}

During the recursive exploration of the search space (i.e., FU-tree) with the total order $\prec$, the partial anti-monotonicity of HUOPs is used to effectively avoid generating unpromising patterns that cannot be HUOPs and their child nodes. Based on the proposed pruning strategies, the upper bound can be used to prune patterns with low support or low utility occupancy early, without constructing their UO-lists. It can effectively reduce both the computational cost of join operations and the search space in the FU-tree. In fact, the actual search space is more compact than a complete FU-tree.

Details of the \textbf{\textit{HUOP-Search}} procedure is shown in Algorithm 3. It takes as input an itemset $X$, a set of UO-lists of all 1-extensions of $X$, $ \alpha $ and $ \beta $. A loop is first performed over each 1-extension of itemset $ X $, denoted as $X_{a}$ (Lines 1 to 20). Two conditions are checked (Lines 3 to 4, Strategy 1) to determine if $X_{a}$ is a HUOP. If the itemset is a HUOP, then it is put into the set of HUOPs. The algorithm then determines whether extensions of $X_{a}$ should be explored (Lines 7 to 8, Strategy 2). It uses the function \textit{\textbf{UpperBound}} (c.f. Algorithm 4) to calculate the upper bound of utility occupancy for the subtree rooted at pattern $X_{a}$. If $ \hat{\phi}(X_a) \geq \beta $, then any extensions of $X_{a}$ may be a HUOP. In that case, the depth-first search is performed (Lines 8 to 18). 
The construction procedure  \textit{\textbf{Construct(X, $X_a$, $X_b$)}} is then executed for each 1-extension node $X_{ab}$, to construct their UO-lists and FU-tables (Lines 11 to 12). Note that each constructed itemset $X_{ab}$ is an 1-extension of the itemset $X_a$. If $ X_{ab}.UOL \not= \emptyset $ and $ sup(X_{ab}) \geq \alpha \times |D| $, such itemset is added to a set  $\textit{extenOfX}_a $ for storing all 1-extensions of $X_a$ (Lines 13 to 15, Strategy 4). HUOPM recursively applies the designed \textbf{\textit{HUOP-Search}} function for determining the desired HUOPs (Line 17).

\begin{algorithm}
	\label{UpperBound procedure}
	\caption{UpperBound ($ X_q.UOL $, $ \alpha $)}
	\begin{algorithmic}[1]				
		\STATE $ sumTopK \leftarrow 0, \hat{\phi}(X_a) \leftarrow 0, V_{occu} \leftarrow \emptyset $;
		\STATE calculate $ (uo(X,T_q) + ruo(X,T_q)) $ of each tuple from the built $ X_{a}.UOL $ and put them into the set of $ V_{occu} $;
		\STATE sort $ V_{occu} $ by descending order as $ V_{occu}^{\downarrow} $;
		
		\FOR {$ k \leftarrow $ 1 to  $\alpha \times |D| $ in $ V_{occu}^{\downarrow} $}		
		\STATE $ sumTopK \leftarrow  sumTopK +  V_{occu}^{\downarrow}[k]  $;					
		\ENDFOR
		
		\STATE $ \hat{\phi}(X_a) = \dfrac{sumTopK}{\alpha \times |D|} $.	
		
		\STATE \textbf{return} $ \hat{\phi}(X_a) $
	\end{algorithmic}
\end{algorithm}

Algorithm 4 provides the pseudo-code of \textit{\textbf{UpperBound}} procedure, for computing the upper bound of utility occupancy. The complexity of Algorithm 4 is $ O(n \times log(n) \times V_{occu}^{\downarrow}[k]) $ (the top $k$ value in vector  is equal to $ \alpha \times |D| $). We first use $ O(n \times log(n)) $  time for sorting $ V_{occu} $ with quick sort algorithm in Line 3, then use $ O(n \times log(n) \times V_{occu}^{\downarrow}[k]) $ time for the loop operation from Lines 4 to 7 to calculate the top $k$ values in $ V_{occu}^{\downarrow}[k]) $  (Lines 4 to 5). Finally, it returns the average value of the summation as the upper bound on utility occupancy of the subtree rooted at pattern $X_a$ (Lines 4 to 5).

\section{Experiments} 
\label{sec:experiments}

In this section, extensive experiments are conducted on both real-world and synthetic datasets to evaluate the effectiveness and efficiency of the proposed HUOPM algorithm. The task of the addressed problem is to discover high utility occupancy patterns (HUOPs) with frequency and utility occupancy. Note that  only two prior studies, DOFIA \cite{tang2012incorporating} and OCEAN \cite{shen2016ocean} are closely related to our work. Major differences are that DOFIA considers both the frequency and occupancy for mining high qualified patterns, but the occupancy is based on the number of items in patterns or transactions rather than the concept of utility  \cite{tang2012incorporating}. Thus, in the following, the state-of-the-art OCEAN algorithm \cite{shen2016ocean} is implemented to generate high utility occupancy patterns.  HUOPs* are derived by OCEAN, and  HUOPs are generated by the proposed HUOPM algorithm. 

We first evaluate the mining results by comparing HUOPs* and HUOPs. Then the runtime and the number of visited nodes in the search tree of the HUOPM algorithm are compared. At last, the effect of the processing order of items is also compared and evaluated.

\subsection{Experimental Setup and Datasets}
All algorithms in the experiments are implemented using the Java language and executed on a PC with an Intel Core i5-3460 3.2 GHz processor and 4 GB of memory, running on the 32 bit Microsoft Windows 7 platform. Four real-world datasets \cite{fimdatasets} (BMSPOS2, retail, chess, and mushroom) and two synthetic dataset  \cite{IBMdata} (T10I4D100K and T40I10D100K) are used in the experiments. T10I4D100K and T40I10D100K are generated using the IBM Quest Synthetic Data Generator \cite{IBMdata}. These datasets have varied characteristics and represents the main types of data typically encountered in real-life scenarios (e.g., dense, sparse and long transactions). The characteristics of these datasets are described below in details.

\begin{itemize}
	\item  \emph{BMSPOS2} dataset contains several years worth of point-of-sale data from a large electronics retailer. It has total 515,597 transactions with 1,657 distinct items, an average transaction length of 6.53 items.
	\item \emph{retail} is a sparse dataset which contains 88,162 transactions with 16,470  distinct items and an average transaction length of 10.30 items.
	\item  \emph{chess} dataset contains 3,196 transactions with 75 distinct items and an average transaction length of 36 items. It is a dense dataset.
	\item  \emph{mushroom} is a dense dataset, it has 8,124 transactions with 120 distinct items, and an average transaction length of 23 items.
	\item  \emph{T10I4D100K} is a synthetic dataset that contains 870 distinct items, 100,000 transactions, and has an average length of 10.1 items.
	\item  \emph{T40I10D100K} is a synthetic dataset that contains 942 distinct items, 100,000 transactions, and has an average length of 39.6 items.
\end{itemize}


For the addressed utility-based HUOPM problem, note that the quantity and unit profit of each item in a dataset is randomly generated by using a simulation method proposed in previous studies \cite{liu2012mining,tseng2010up,tseng2013efficient}. In the following experiments, the proposed HUOPM algorithm with different designed pruning strategies is respectively denoted as HUOPM$ _{P12} $ (with pruning strategies 1 and 2), HUOPM$ _{P13} $ (with pruning strategies 1 and 3), HUOPM$ _{P123} $ (with pruning strategies 1, 2 and 3) and HUOPM$ _{P1234} $ (with all pruning strategies). The four versions are compared to evaluate the efficiency of HUOPM and the effect of the designed pruning strategies.

\subsection{Pattern Analysis}
In the study presenting OCEAN, it was shown that the diversity of patterns with high utility occupancy is more favorable than traditional high utility patterns and frequent patterns. The study used a small real-life dataset \textit{MobileApp} to illustrate that the derived patterns are more useful for the task of mobile App recommendation or promotion \cite{shen2016ocean}.  To analyze the usefulness of the proposed HUOPM framework, the derived patterns (HUOPs* and HUOPs) are evaluated below. When setting $ \alpha $:20\% and $ \beta $:0.50 on the running example, the top-10 patterns derived by HUOPM, OCEAN and FP-growth \cite{han2004mining} are shown in TABLE \ref{table:example}, respectively. Notice that the results in left, middle and right of TABLE \ref{table:example} are sorted by the descending order of \textit{utility occupancy} (\textit{uo}) and support (\textit{sup}), respectively. It is clearly that three real HUOPs of top-10 results ($(abde)$, $(abe)$, and $(cd)$, as shown in TABLE \ref{table:example} (left)) are missed in the results of OCEAN (as shown in TABLE \ref{table:example} (middle)). Besides, the minimum \textit{utility occupancy} in top-10 HUOPs* is 0.5604, while the minimum \textit{utility occupancy} value in top-10 HUOPs is 0.6554. Some interesting desired patterns with high \textit{uo} values cannot be discovered by OCEAN. Therefore, the results derived by the OCEAN algorithm are incomplete and OCEAN also encounters some performance problems (it will be discussed in the next subsection). Finally, it can be clearly seen that the support-based frequent patterns (as shown in TABLE \ref{table:example} (right)) are quite different from the results derived by the utility-occupancy-based methods.

\begin{table}[htb]
	\fontsize{5.5pt}{9pt}\selectfont
	\centering
	\caption{Derived patterns from running example}
	\label{table:example}
	\begin{tabular}{lll||lll||lll}
		\hline\hline
		\multicolumn{3}{c}{(\textbf{left}) Top-10 HUOPs} &
		\multicolumn{3}{c}{(\textbf{middle}) Top-10 HUOPs*} &
		\multicolumn{3}{c}{(\textbf{right}) Top-10 FPs} \\
		\hline
		\textbf{Pattern} &  \textit{\textbf{sup}} &  \textit{\textbf{uo}${\downarrow}$}  & \textbf{Pattern} &  \textit{\textbf{sup}} &  \textit{\textbf{uo}${\downarrow}$} & \textbf{Pattern} &  \textit{\textbf{sup}${\downarrow}$} &  \textit{\textbf{uo}}   \\ \hline
		
		$(abcd)$ & 2 & 0.9081 & $(abcd)$ & 2 & 0.9081 & $(c)$ & 8 & 0.6468 \\
		$(acd)$ & 4 & 0.8972 & $(acd)$ & 4 & 0.8972  & $(d)$ &　7 & 0.0897\\
		$(ce)$ & 3 & 0.8776 & $(ce)$ & 3 & 0.8776　 & $(a)$ &5  & 0.2985 \\
		$(abc)$ & 2 & 0.8308 & $(abc)$ & 2 & 0.8308  & $(b)$ & 5 & 0.2192 　\\
		$(ac)$ & 4 & 0.8273 & $(ac)$ & 4 & 0.8273  & $(ad)$ & 5 &　0.3609 \\
		$(abde)$ & 2 & 0.7755 & 	$(ade)$ & 2 & 0.6979   & $(cd)$ & 5 & 0.6881　\\
		$(abe)$ & 2 & 0.7081 & $(bc)$ & 3 & 0.6554	 & $(e)$ & 4 &　0.4022 \\	
		$(ade)$ & 2 & 0.6979 & $(c)$ & 8 & 0.6468	  & $(ac)$ & 4 &　0.8273 \\	
		$(cd)$ & 5 & 0.6881 & $(ae)$ & 2 & 0.6305	 & $(bd)$ & 4 &　0.3620 \\
		$(bc)$ & 3 & 0.6554 & $(bcd)$ & 2 & 0.5604  & $(acd)$ & 4 & 0.8972　\\ 		
		\hline\hline
	\end{tabular}
\end{table}

We further compare the number of patterns on the test datasets to show the effect of varying the parameters  $ \alpha $  and  $ \beta $. Results for various parameter values are shown in TABLE \ref{table:patterns1} and TABLE \ref{table:patterns2}, respectively. Notice that the parameters in TABLE \ref{table:patterns1} are varying $ \alpha $ under a fixed $\beta$ = 0.3, and the parameter settings on each dataset are the same as them in Fig. \ref{fig:Runtime1}. In TABLE \ref{table:patterns2}, the parameters are varying $ \beta $ under a fixed $\alpha$, and details of the parameter settings in each dataset are the same as them in Fig. \ref{fig:Runtime2}.

\begin{table}[htb]
	\fontsize{5.5pt}{9pt}\selectfont
	\centering
	\caption{Derived patterns under varied $ \alpha $}
	\label{table:patterns1}
	\begin{tabular}{c|c|llllll}
		\hline\hline
		\multirow{2}*{\textbf{Dataset}}&
		\multirow{2}*{\textbf{Patterns}}
		&\multicolumn{6}{c}{\# \textbf{patterns by varying threshold $ \alpha $}}\\
		\cline{3-8}
		&&$ \alpha_1 $ & $ \alpha_2 $ & $ \alpha_3 $ & $ \alpha_4 $ &  $ \alpha_5 $ &  $ \alpha_6 $ \\ \hline
		BMSPOS2 &HUOPs* &29,389 &17,005 &11,876 &8,961 &7,065 &5,884 \\
		($\beta$: 0.3)&HUOPs& 29,444 &17,048 &11,899 &8,987 &7,094 &5,905 \\
		\hline
		retail&HUOPs* &12,801 &8,155 &5,902 &4,513 &3,989 &3,282 \\
		($\beta$: 0.3)&HUOPs& 13,105 &8,209 &5,927 &4,527 &4,004 &3,291 \\
		\hline
		chess&HUOPs* &2,423 &1,776 &1,374 &1,108 &903 &695 \\
		($\beta$: 0.3)&HUOPs& 11,469 &8,949 &7,075 &5,510 &4,273 &3,351 \\
		\hline
		mushroom&HUOPs* &17,678 &12,429 &10,739 &10,252 &2,821 &1,199 \\
		($\beta$: 0.3)&HUOPs& 52,795 &33,629 &31,741 &31,025 &4,613 &1,672 \\
		\hline	
		T10I4D100K&HUOPs* &536,414 &246,932 &138,877 &91,557 &67,834 &53,957 \\
		($\beta$: 0.3)&HUOPs& 550,676 &253,617 &144,106 &96,505 &72,555 &59,068 \\
		\hline
		T40I10D100K&HUOPs* &33,171 &14,415 &10,465 &258 &151 &55 \\
		($\beta$: 0.3)&HUOPs& 71,941 &15,496 &12,273 &1,367 &1,002 &548 \\	
		\hline\hline
	\end{tabular}
\end{table}

\begin{table}[htb]
	\fontsize{5.5pt}{9pt}\selectfont
	\centering
	\caption{Derived patterns under varied $ \beta $}
	\label{table:patterns2}
	\begin{tabular}{c|c|llllll}
		\hline\hline
		\multirow{2}*{\textbf{Dataset}}&
		\multirow{2}*{\textbf{Patterns}}
		&\multicolumn{6}{c}{\# \textbf{patterns by varying threshold $ \beta $}}\\
		\cline{3-8}
		&&$ \beta_1 $ & $ \beta_2 $ & $ \beta_3 $ & $ \beta_4 $ &  $ \beta_5 $ &  $ \beta_6 $ \\ \hline
		BMSPOS2 &HUOPs* &287,850 &106,239 &34,195 &11,876 &5,089 &2,696 \\
		($\alpha$: 0.02\%)&HUOPs& 287,902 &106,285 &34,233 &11,899 &5,115 &2,721 \\
		\hline
		retail&HUOPs* &31,036 &13,910 &5,902 &2,295 &842 &315 \\
		($\alpha$: 0.014\%)&HUOPs& 31,067 &13,940 &5,927 &2,335 &883 &346 \\
		\hline
		chess&HUOPs* &9,820 &4,837 &1,374 &514 &125 &14 \\
		($\alpha$: 62\%)&HUOPs& 24,905 &14,050 &7,075 &3,153 &1,170 &357 \\
		\hline
		
		mushroom&HUOPs* &16,733 & 10,739 & 4,914 & 1,327 & 505 & 136\\
		
		($\alpha$: 18\%)&HUOPs& 41,311 &31,741 &21,755 &12,917 &6,431 &2,556 \\
		\hline	
	
		T10I4D100K&HUOPs* &181,541 &	131,631 &	91,557 &	61,079 &	38,980 &	23,435\\
		
		($\alpha$: 0.012\%)&HUOPs& 182,538 &135,710 &96,505 &66,772 &44,380 &27,906 \\
		\hline
		T40I10D100K&HUOPs* &60,749 &37,163 &21,105 &10,465 &3,898 &911 \\
		($\alpha$: 0.8\%)&HUOPs& 68,457 &42,887 &24,592 &12,273 &5,031 &1,594 \\	
		\hline\hline
	\end{tabular}
\end{table}

From TABLE \ref{table:patterns1} and TABLE \ref{table:patterns2}, it can be clearly observed that sets of derived patterns are quite different for various $ \alpha $ or $ \beta $ values. The number of HUOPs* is always smaller than that of HUOPs. For example, on chess dataset as shown in TABLE \ref{table:patterns1} and TABLE \ref{table:patterns2}, most of final HUOPs are missed by the OCEAN algorithm since the number of HUOPs* is always considerably smaller  than that of HUOPs. It means that numerous interesting high utility occupancy patterns are effectively discovered by HUOPM algorithm, while most of them are missed by the OCEAN algorithm. In other words, although the addressed HUOPM problem can capture  high qualified patterns well, OCEAN fails to discover the complete set of HUOPs. It can also be observed that the number of  produced patterns (both HUOPs* and HUOPs), decrease when the minimum support threshold is increased. And less HUOPs* and HUOPs are obtained when $ \beta $ is set higher. Besides, the number of missing patterns which is caused by OCEAN (i.e., HUOPs - HUOPs*) sometimes increases when varying $ \beta $, while it sometimes decreases, as shown on the chess and T10I4D100K datasets when varying $ \beta $ with a fixed $ \alpha $. From the above analysis of the results of found patterns, it can be concluded that the proposed HUOPM algorithm for mining HUOPs is acceptable and can solve a serious shortcoming of the state-of-the-art OCEAN algorithm.

\subsection{Efficiency Analysis}
To evaluate the efficiency of OCEAN and the proposed HUOPM algorithm on various datasets, a performance comparison of the different strategies used in HUOPM is presented next in terms of execution time. The results in terms of runtime for various parameters are shown in Fig. \ref{fig:Runtime1} and Fig. \ref{fig:Runtime2}, respectively.

\begin{figure}[htbp]
	\centering
    \includegraphics[trim=30 13 45 25,clip,scale=0.46]{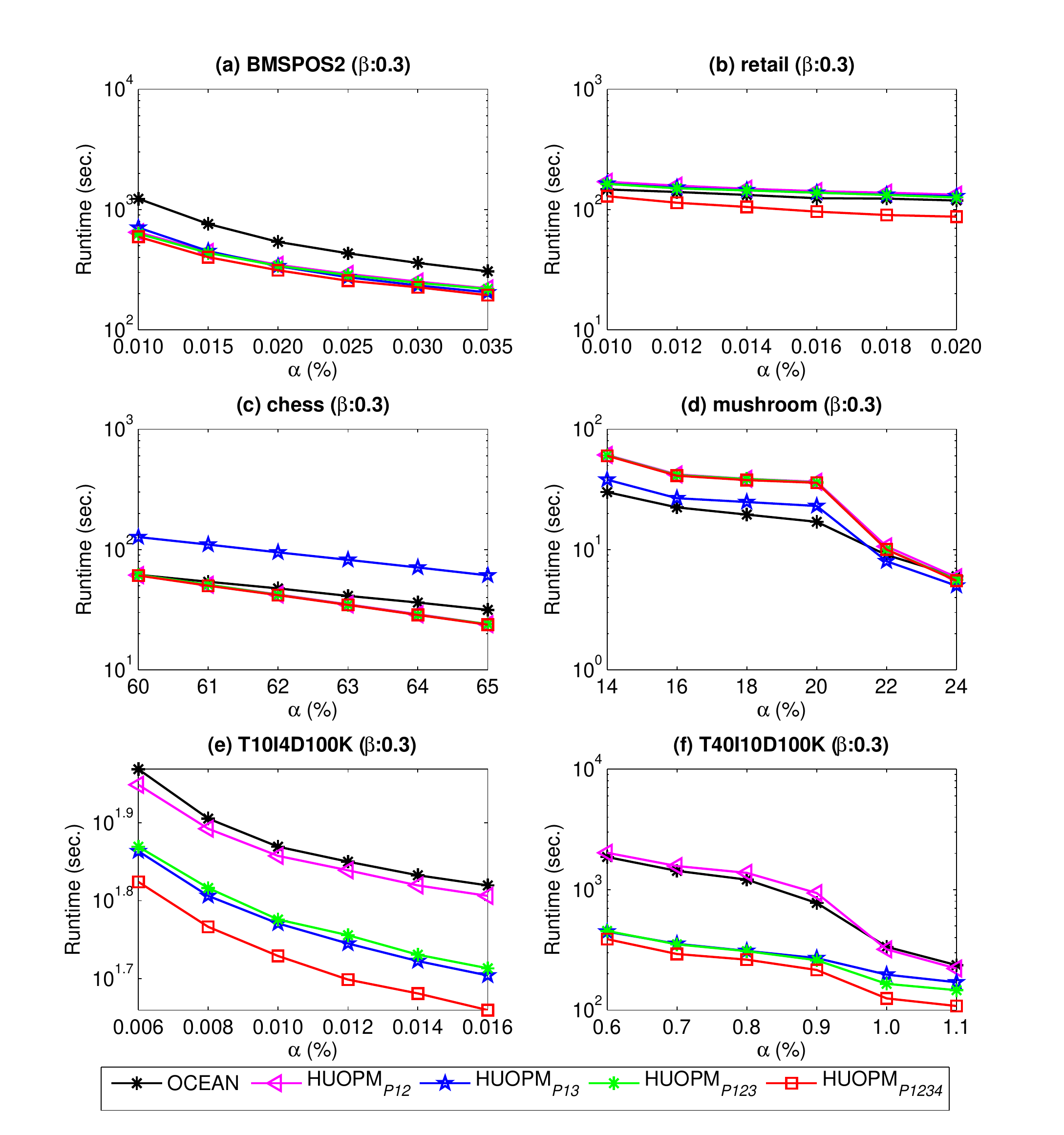}
	\caption{Runtime under varied $ \alpha $ with a fixed $ \beta $.}
	\label{fig:Runtime1}	
\end{figure}

\begin{figure}[htbp]
	\centering
    \includegraphics[trim=30 13 45 25,clip,scale=0.46]{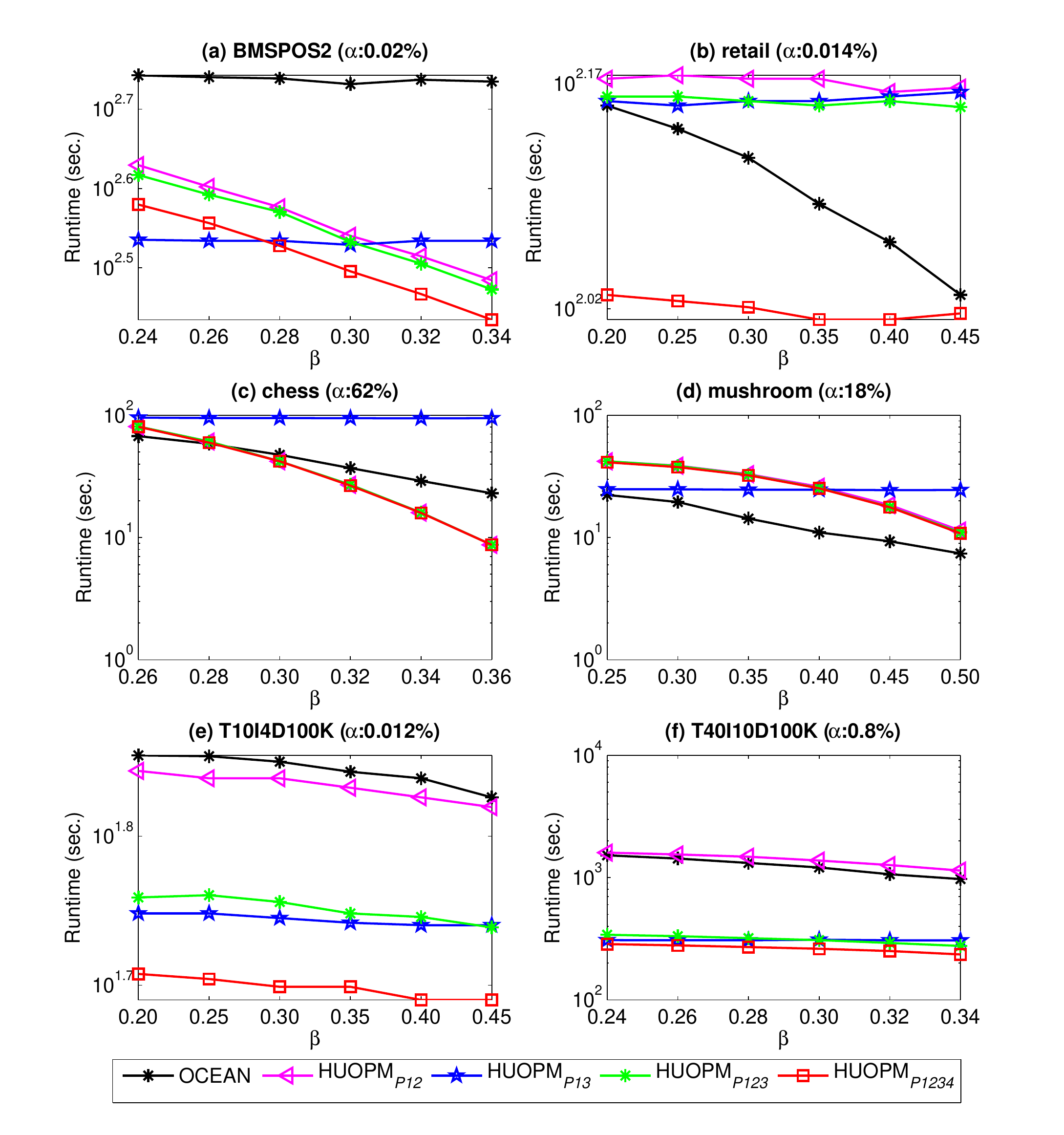}
	\caption{Runtime under varied $ \beta $ with a fixed $ \alpha $.}
	\label{fig:Runtime2}	
\end{figure}

From Fig. \ref{fig:Runtime1} and Fig. \ref{fig:Runtime2}, it can be clearly observed that the runtime of OCEAN is the worst, compared to the other algorithms in most cases, and that the HUOPM$ _{P1234} $ algorithm (which adopts  all pruning strategies) is the fastest on all datasets in all cases. 
We also draw the following conclusions. 
(1) The difference between HUOPM$ _{P12} $ and HUOPM$ _{P123} $ indicates that the Strategy 3 always reduces the search space by pruning subtrees, as it can be observed in Fig. \ref{fig:Runtime1} and Fig. \ref{fig:Runtime2}. 
(2) By comparing HUOPM$ _{P13} $ and HUOPM$ _{P123} $, it can be concluded that the Strategy 2, using the upper bound of utility occupancy, provides a trade-off between efficiency and effectiveness. 
For example, consider results for BMSPOS2, shown in Fig. \ref{fig:Runtime2}(a). When $ \beta $ = 0.24, HUOPM$ _{P13} $ is faster than HUOPM$ _{P123} $; while HUOPM$ _{P13} $ consumes more runtime than HUOPM$ _{P123} $ when $ \beta $ = 0.32 and  $ \beta $ = 0.34. The reason for this behavior is that it needs to spend additional time to calculate the upper bounds, but sometimes unpromising itemsets can be directly pruned by other pruning strategies. (3) In general, the Strategy 4 can significantly reduce the execution time. For example, as shown in Fig. \ref{fig:Runtime1}(b), (d) and (f), Fig. \ref{fig:Runtime2}(a), (b) and (e), the runtime of HUOPM$ _{P1234} $ is considerably smaller  than that of HUOPM$ _{P123} $. Therefore, the HUOPM algorithm applies four pruning strategies to prune unpromising patterns early, which greatly speed up the mining efficiency, compared to the baseline algorithms and the OCEAN algorithm.

Furthermore, we evaluate the statistical analysis \cite{demvsar2006statistical} by two-way ANOVA test \cite{wonnacott1990introductory} to show whether the proposed HUOPM has a significant difference to the traditional OCEAN algorithm. Based on a two-way ANOVA analysis, in most cases, there is a significant difference of the runtime between HUOPM and OCEAN under varied $ \alpha $ with a fixed $ \beta $. Details are described below: $F$ = 11.551, $P$ = 0.019 $<$ 0.050, in Fig. \ref{fig:Runtime1}(a); $F$ = 156.372, $P$ $<$ 0.001, in Fig. \ref{fig:Runtime1}(b); $F$ = 3.619, $P$ = 0.116, in Fig. \ref{fig:Runtime1}(c); $F$ = 8.871, $P$ = 0.031 $<$ 0.050, in Fig. \ref{fig:Runtime1}(d); $F$ = 491.368, $P$ $<$ 0.001, in Fig. \ref{fig:Runtime1}(e); $F$ = 11.597, $P$ = 0.019 $<$ 0.050, in Fig. \ref{fig:Runtime1}(f). Consider Fig. \ref{fig:Runtime2}, the results of a two-way ANOVA analysis are presented below. $F$ = 211.304, $P$ $<$ 0.001, in Fig. \ref{fig:Runtime2}(a); $F$ = 18.649, $P$ = 0.008 $<$ 0.010, in Fig. \ref{fig:Runtime2}(b); $F$ = 1.303, $P$ = 0.305, in Fig. \ref{fig:Runtime2}(c); $F$ = 27.371, $P$ = 0.003 $<$ 0.050, in Fig. \ref{fig:Runtime2}(d); $F$ = 11705.343, $P$ $<$ 0.001, in Fig. \ref{fig:Runtime2}(e); $F$ = 157.518, $P$ $<$ 0.001, in Fig. \ref{fig:Runtime2}(f). Thus, there is a significant difference of the runtime between the proposed HUOPM and the OCEAN algorithm under different parameter settings.

\subsection{Effect of Pruning Strategies} 
To assess the impact of pruning strategies, we further compared the number of nodes visited in the FU-tree. The number of nodes visited by the four versions of the proposed algorithm (HUOPM$ _{P12} $, HUOPM$ _{P13} $, HUOPM$ _{P123} $ and HUOPM$ _{P1234} $) are denoted as $N_1$, $N_2$, $N_3$ and $N_4$, respectively. Results for the same parameter settings as in previous experiments are shown in Fig. \ref{fig:Nodes1} and Fig. \ref{fig:Nodes2}, respectively.

\begin{figure}[htbp]
	\centering
	\includegraphics[trim=35 20 45 25,clip,scale=0.46]{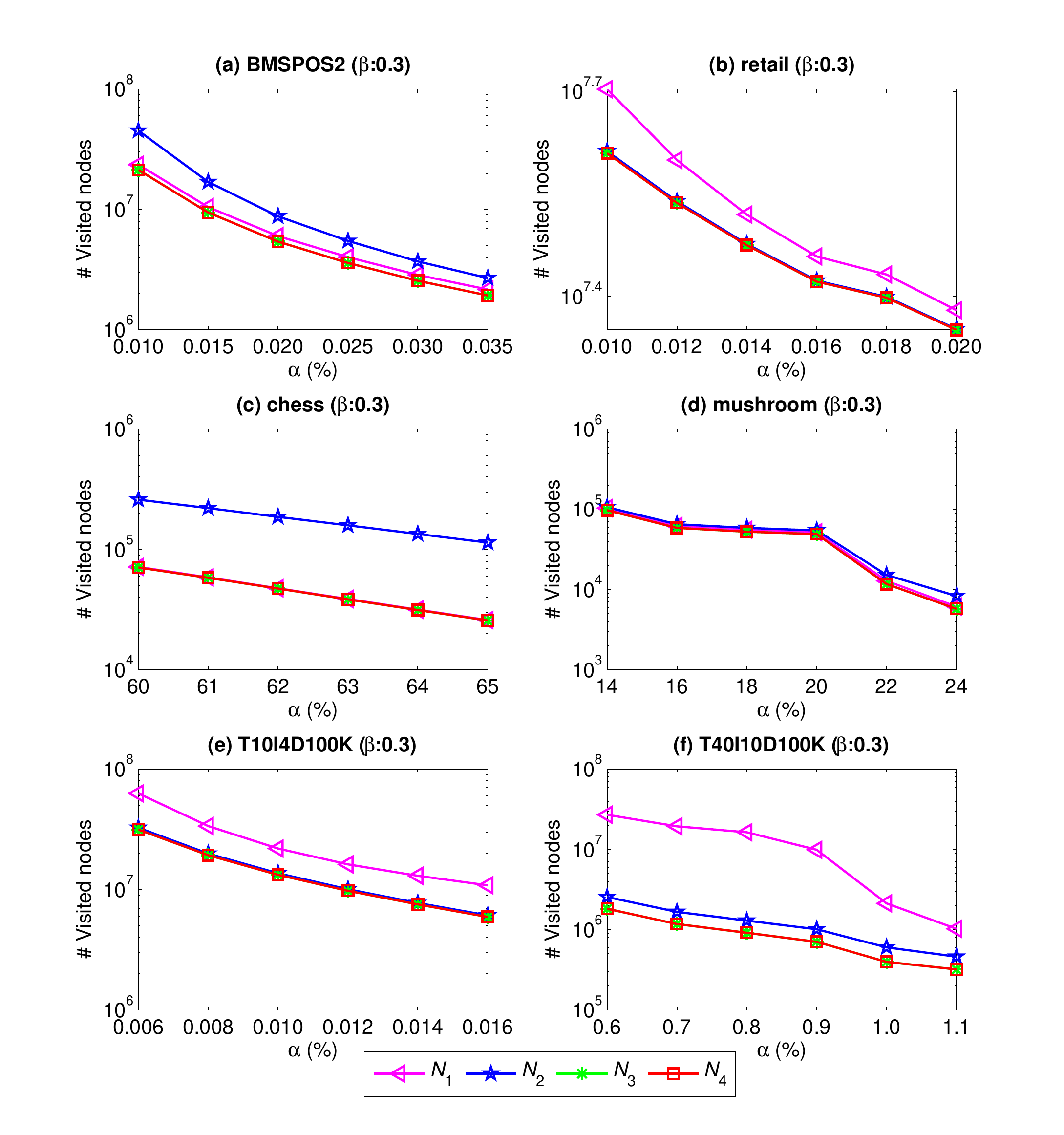}
	\caption{The number of visited nodes under varied $ \alpha $ with a fixed $ \beta $.}
	\label{fig:Nodes1}	
\end{figure}

\begin{figure}[htbp]
	\centering
	\includegraphics[trim=35 18 45 20,clip,scale=0.46]{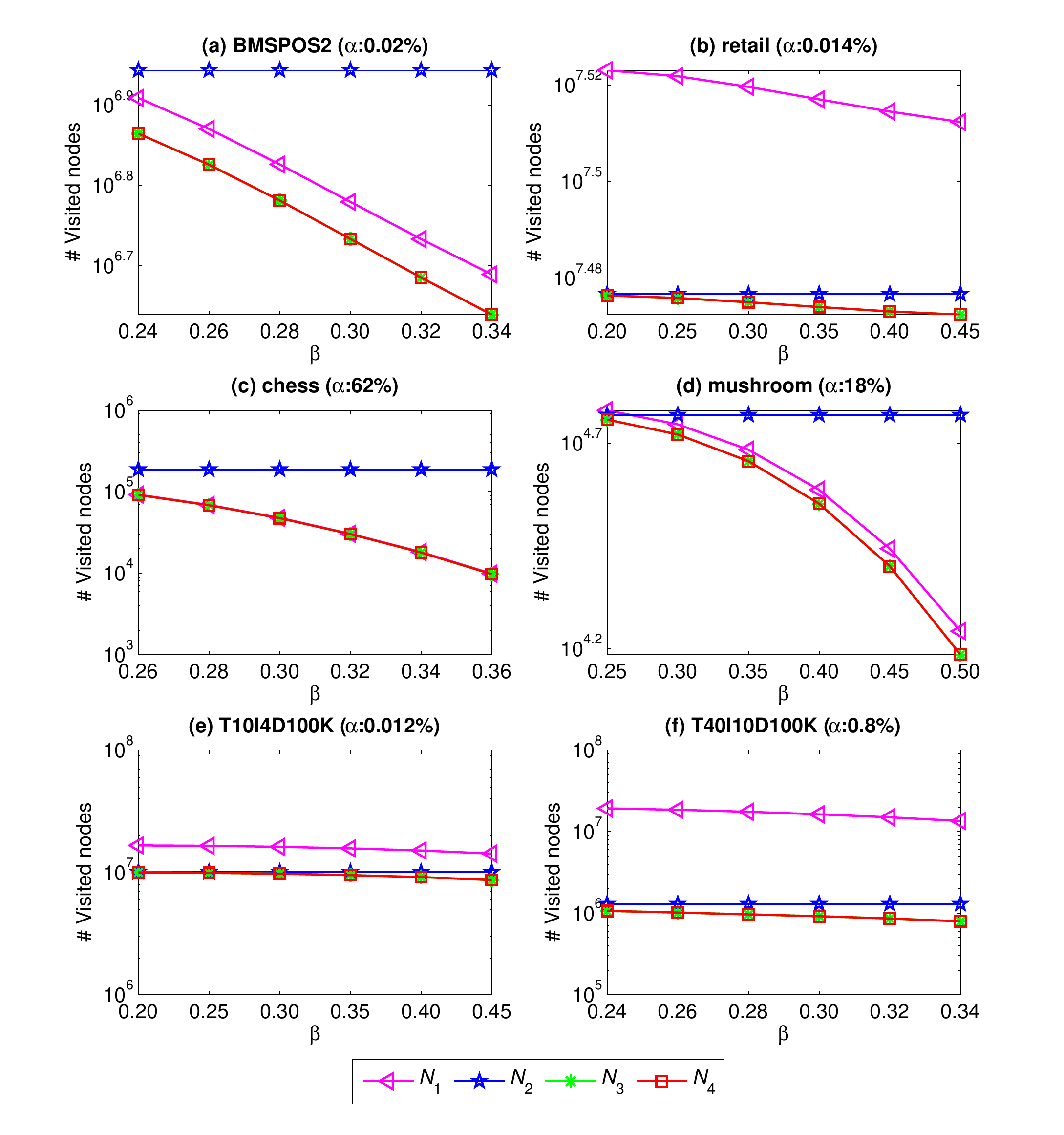}
	\caption{The number of visited nodes under varied $ \beta $ with a fixed $ \alpha $.}
	\label{fig:Nodes2}	
\end{figure}

Based on Fig. \ref{fig:Runtime1} to Fig. \ref{fig:Nodes2}, five observations are made. 
(1) HUOPM prunes a larger part of the search space in terms of number of nodes and is faster when  $ \alpha $ and $ \beta $ are increased. 
In general, the runtime and the number of visited nodes in the search tree of HUOPM decrease as $ \alpha $ or $ \beta $ increases. 
Therefore, the effect of the pruning strategies on the number of visited nodes, to some extent, reflects the execution time required by the algorithms. Hence, a smaller tree is more quickly processed since less execution time is required for spanning the FU-tree and constructing UO-lists.
(2) Comparing $ N_2 $ and $ N_3 $ (with and without the upper bound) indicates that the proposed upper bound is useful for pruning unpromising patterns, thus leading to a more compact search space on both sparse and dense datasets. (3) The numbers  $ N_3 $ and  $ N_4 $ are always the same regardless of the increase in $ \alpha $ and $ \beta $. It indicates that the Strategy 4 does not  filter more unpromising patterns when compared to Strategy 3, but it has a powerful effect on reducing the runtime. We can see this trend more clearly in Fig. \ref{fig:Runtime1} and Fig. \ref{fig:Runtime2} by observing the difference between HUOPM$ _{P123} $ and HUOPM$ _{P1234} $. Thus, HUOPM$ _{P1234} $ always consumes the least runtime among the four versions of the algorithm. (4) The higher $ \alpha $ and $ \beta $ are set, the smaller the search space (number of visited nodes) is. Thus, the less runtime is required. Moreover, the less memory usage is required, but we omit the detailed memory consumption results due to the page limit. (5) The utility occupancy is a significant factor for improving mining performance in the addressed mining task.

\subsection{Processing Order of Items}
Since the processing order of items may influence the performance of a mining algorithm, it is important to select a suitable sorting order for the proposed algorithm. To assess how different processing orders influence the performance of the proposed HUOPM algorithm, we measured the runtime and memory consumption of HUOPM with all pruning strategies but with different processing orders. Five types of processing orders are evaluated: the lexicographic order (denoted as HUOPM$ _{lexi} $), the transaction-weighted utilization ascending order (denoted as HUOPM$ _{twuas} $), the transaction-weighted utilization descending order (denoted as HUOPM$ _{twude} $), the support ascending order (denoted as HUOPM$ _{supas} $), and the support ascending descending order (denoted as HUOPM$ _{supde} $).

\begin{figure}[!htbp]
	\centering
	\includegraphics[trim=105 10 80 0,clip,scale=0.3]{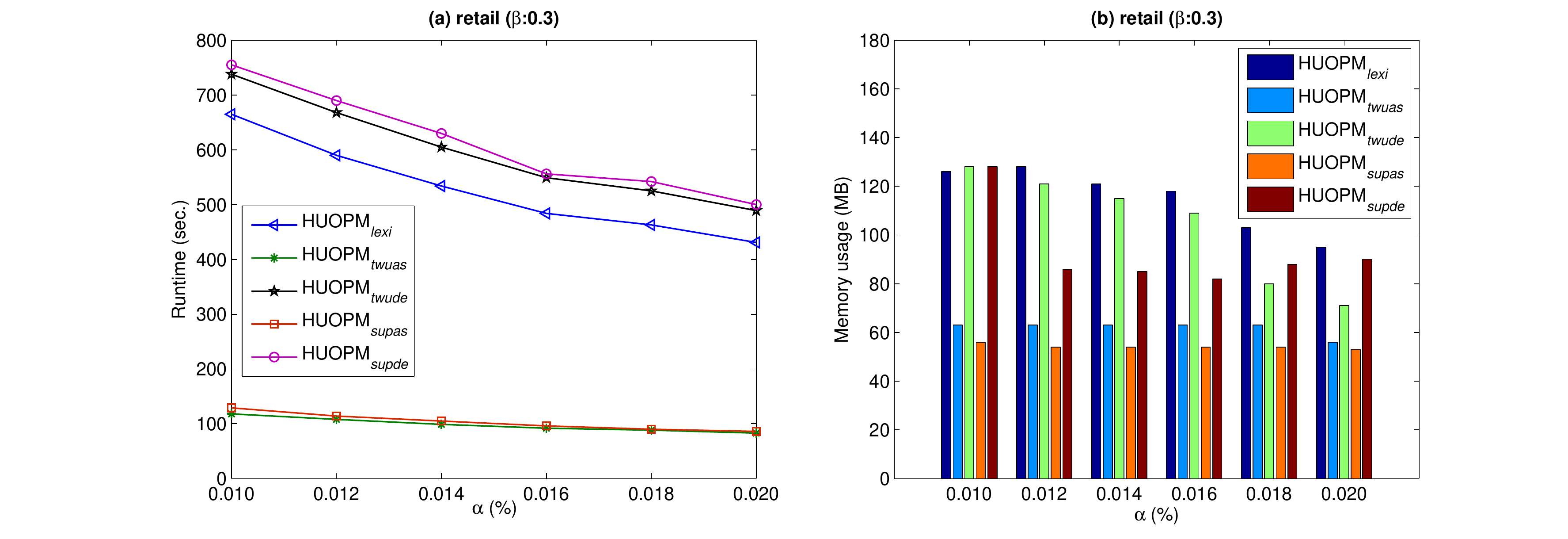}
	\caption{Effect of processing order of items (Five types of processing orders: HUOPM$ _{lexi} $, HUOPM$ _{twuas} $, HUOPM$ _{twude} $, HUOPM$ _{supas} $ and HUOPM$ _{supde} $).}
	\label{fig:OrderOfItems}
\end{figure}

Fig. \ref{fig:OrderOfItems} shows the experimental results on the accidents and retail datasets, respectively. As we can see, HUOPM$ _{twuas} $ has  similar performance when compared to HUOPM$ _{supas} $. However, HUOPM$ _{twuas} $ requires more memory usage than HUOPM$ _{supas} $, as shown in Fig. \ref{fig:OrderOfItems}(b). Thus, the adopted support ascending order (HUOPM$ _{supas} $) always leads to the best performance in terms of execution time and memory usage. Clearly, the \textit{TWU} or support descending order provides the worst performance. It indicates that the adopted processing order w.r.t. support ascending order of items in the proposed HUOPM approach can greatly reduce the number of UO-lists and  join operations performed for a mining task.

\section{Conclusions and Future Work} 
\label{sec:conclusion}

In this paper, we propose an effective and efficient HUOPM algorithm to address a new research problem of mining high utility occupancy patterns with utility occupancy. The utility occupancy can lead to useful patterns that contribute a large portion of total utility for each individual transaction representing user interests or user habit. To the best of our knowledge, no prior algorithms address this problem successfully and effectively. It may provide a new research perspective for utility mining, to a certain extent. In HUOPM, an upper bound of utility occupancy for all possible extensions rooted at a processed node/itemset can be quickly derived by utilizing the UO-list. Based on the developed two \textit{downward closure} properties and four pruning strategies, the HUOPM algorithm can directly discover the profitable HUOPs from the FU-tree using UO-list without candidate generation. Extensive experiments on several datasets show that HUOPM can efficiently find the complete set of high utility occupancy patterns and significantly outperforms the state-of-the-art OCEAN algorithm. Moreover, the proposed pruning strategies are powerful to prune the search space and speed up the mining performance. Extending HUOPM to address the distributed mining problem \cite{gan2017data}, dynamic mining \cite{2gan2018survey}, and privacy preserving  issue \cite{gan2018privacy} are part of our future work.

\ifCLASSOPTIONcaptionsoff
  \newpage
\fi

\bibliographystyle{IEEEtran}
\bibliography{paper}


\vspace{-1.2cm}
\begin{IEEEbiography}[{\includegraphics[width=1in,height=1.25in,clip,keepaspectratio]{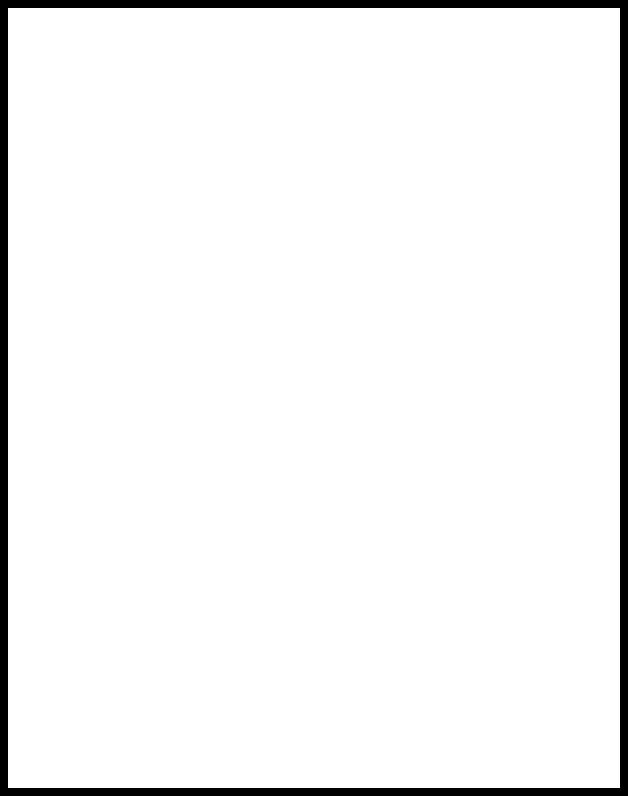}}]{Wensheng Gan}
	is currently a joint Ph.D. candidate at the Department of Computer Science, University of Illinois at Chicago, IL, USA; and a Ph.D. candidate at the School of Computer Science and Technology, Harbin Institute of Technology (Shenzhen), Guangdong, China. He received B.S. and M.S. degrees in Computer Science from South China Normal University and Harbin Institute of Technology (Shenzhen), Guangdong, China, in 2013 and 2015, respectively. His research interests include data mining, utility mining and Big Data technologies. He has published more than 50 research papers in peer-reviewed journals and international conferences.
\end{IEEEbiography}

\vspace{-1.2cm}
\begin{IEEEbiography}[{\includegraphics[width=1in,height=1.25in,clip,keepaspectratio]{newAuthor.png}}]{Jerry Chun-Wei Lin}
	is an associate professor at Department of Computing, Mathematics and Physics, Western Norway University of Applied Sciences, Bergen, Noway. He received the Ph.D. in Computer Science and Information Engineering, National Cheng Kung University, Tainan, Taiwan in 2010. His research interests include data mining, privacy-preserving and security, big data analytics, and machine learning.  He has published more than 250 research papers in peer-reviewed international conferences and journals. He is the co-leader of the popular SPMF open-source data mining library, the funder of the PPSF open-source privacy-preserving and security project, and Editor-in-Chief of the \textit{Data Mining and Pattern Recognition} (DSPR) journal.
\end{IEEEbiography}

\vspace{-2cm}
\begin{IEEEbiography}[{\includegraphics[width=1in,height=1.25in,clip,keepaspectratio]{newAuthor.png}}]{Philippe Fournier-Viger}
	is a full professor and Youth 1000 scholar at the Harbin Institute of Technology (Shenzhen), Shenzhen, China. He received the Ph.D. in  Computer Science from the University of Quebec, Montreal, in 2010. His research interests include data mining, pattern mining, sequence analysis and prediction, e-learning, and social network mining. He has published more than 200 research papers in peer-reviewed international conferences and journals. He is the founder of the popular SPMF open-source data-mining library, which has been cited in more than 700 research papers since 2010. He is also Editor-in-Chief of the \textit{Data Mining and Pattern Recognition} (DSPR) journal.
\end{IEEEbiography}

\vspace{-2cm}
\begin{IEEEbiography}[{\includegraphics[width=1in,height=1.25in,clip,keepaspectratio]{newAuthor.png}}]{Han-Chieh Chao}
	has been the president of National Dong Hwa University since February 2016. He received M.S. and Ph.D. degrees in Electrical Engineering from Purdue University in 1989 and 1993, respectively. His research interests include high-speed networks, wireless networks, IPv6-based networks, and artificial intelligence. He has published nearly 500 peer-reviewed professional research papers. He is the Editor-in-Chief (EiC) of IET Networks' \textit{Journal of Internet Technology}. He is the founding EiC of the \textit{International Journal of Internet Protocol Technology }and the \textit{International Journal of Ad Hoc and Ubiquitous Computing}. Dr. Chao has served as the guest editor for ACM MONET, IEEE JSAC, \textit{IEEE Communications Magazine}, \textit{IEEE Systems Journal}, \textit{Computer Communications}, \textit{IEEE Proceedings Communications}, \textit{Wireless Personal Communications}, and \textit{Wireless Communications \& Mobile Computing}. Dr. Chao is an IEEE Senior Member and a fellow of IET. 
\end{IEEEbiography}

\vspace{-2cm}
\begin{IEEEbiography}[{\includegraphics[width=1in,height=1.25in,clip,keepaspectratio]{newAuthor.png}}]{Philip S. Yu}
	received the B.S. degree in Electrical Engineering from National Taiwan University, M.S. and Ph.D. degrees in Electrical Engineering from Stanford University, and an MBA from New York University. He is a distinguished professor of computer science with the University of Illinois at Chicago (UIC) and also holds the Wexler Chair in Information Technology at UIC. Before joining UIC, he was with IBM, where he was manager of the Software Tools and Techniques Department at the Thomas J. Watson Research Center. His research interests include data mining, data streams, databases, and privacy. He has published more than 1,200 papers in peer-reviewed journals (i.e., TKDE, TKDD, VLDBJ, ACM TIST) and conferences (KDD, ICDE, WWW, AAAI, SIGIR, ICML, CIKM, etc). He holds or has applied for more than 300 U.S. patents. Dr. Yu is the Editor-in-Chief of \textit{ACM Transactions on Knowledge Discovery from Data}. He received the ACM SIGKDD 2016 Innovation Award for his influential research and scientific contributions on mining, fusion, and anonymization of Big Data, and the IEEE Computer Society 2013 Technical Achievement Award. Dr. Yu is a fellow of the ACM and the IEEE.
\end{IEEEbiography}

\end{document}